\pdfoutput=1
\RequirePackage{ifpdf}
\ifpdf % We are running pdfTeX in pdf mode
\documentclass[pdftex]{sigma}
\else
\documentclass{sigma}
\fi

\numberwithin{equation}{section}

\newtheorem{thm}{Theorem}[section]
\newtheorem{cor}[thm]{Corollary}
\newtheorem{lma}[thm]{Lemma}
\newtheorem{lem}[thm]{Lemma}
\newtheorem{prop}[thm]{Proposition}
{\theoremstyle{definition}
\newtheorem{dfn}[thm]{Definition}}

\usepackage{mathtools}
\usepackage[knot,matrix]{xy}

\newcommand{\defeq}{\coloneqq}
\newcommand{\tens}{\otimes}
\newcommand{\ctens}{\hat{\otimes}}
\newcommand{\im}{\mathrm{i}}
\newcommand{\R}{\mathbb{R}}
\newcommand{\C}{\mathbb{C}}
\newcommand{\toi}{\hookrightarrow}
\newcommand{\cH}{\mathcal{H}}
\newcommand{\ds}{\circ}
\newcommand{\cHd}{\mathcal{H}^{\ds}}
\newcommand{\cHl}{\mathcal{H}^\text{L}}
\newcommand{\cHld}{\mathcal{H}^{\text{L}\ds}}
\newcommand{\rhol}{\rho^\text{L}}
\newcommand{\rH}{\mathrm{H}}
\newcommand{\cont}{\mathrm{C}}
\newcommand{\xd}{\mathrm{d}}
\newcommand{\sts}{\mathcal{M}}

\newcommand{\gf}{\varphi}
\newcommand{\geps}{\varepsilon}

\begin{document}
\allowdisplaybreaks

\newcommand{\arXivNumber}{1712.05537}

\renewcommand{\PaperNumber}{105}

\FirstPageHeading

\ShortArticleName{Quantum Abelian Yang--Mills Theory on Riemannian Manifolds with Boundary}

\ArticleName{Quantum Abelian Yang--Mills Theory\\ on Riemannian Manifolds with Boundary}

\Author{Homero G.~D\'IAZ-MAR\'IN~$^\dag$ and Robert OECKL~$^\ddag$}

\AuthorNameForHeading{H.G.~D\'{\i}az-Mar\'{\i}n and R.~Oeckl}

\Address{$^\dag$~Facultad de Ciencias F\'{\i}sico-Matem\'aticas, Universidad Michoacana de San Nicol\'as de Hidalgo,\\
\hphantom{$^\dag$}~Ciudad Universitaria, C.P.~58060, Morelia, Michoac\'an, Mexico}
\EmailD{\href{mailto:hdiaz@umich.mx}{hdiaz@umich.mx}}

\Address{$^\ddag$~Centro de Ciencias Matem\'aticas, Universidad Nacional Aut\'onoma de M\'exico,\\
\hphantom{$^\ddag$}~C.P.~58190, Morelia, Michoac\'an, Mexico}
\EmailD{\href{mailto:robert@matmor.unam.mx}{robert@matmor.unam.mx}}

\ArticleDates{Received December 18, 2017, in final form September 18, 2018; Published online September 27, 2018}

\Abstract{We quantize abelian Yang--Mills theory on Riemannian manifolds with bounda\-ries in any dimension. The quantization proceeds in two steps. First, the classical theory is encoded into an axiomatic form describing solution spaces associated to manifolds. Second, the quantum theory is constructed from the classical axiomatic data in a functorial manner. The target is general boundary quantum field theory, a TQFT-type axiomatic formulation of quantum field theory.}


\Classification{53D30; 58E15; 58E30;81T13}

\section{Introduction}

The present paper represents a step in the program to formalize realistic quantum field theories in terms of the axiomatic approach known as \emph{topological quantum field theory} (TQFT) or \emph{functorial quantum field theory}. More specifically, we show how to quantize abelian Yang--Mills theory functorially on Riemannian manifolds with boundary.

The precise framework we use is \emph{general boundary quantum field theory} (GBQFT) \cite{Oe:gbqft}. We choose this framework over other TQFT-type axiomatizations for two reasons. On the one hand it is more amenable to support infinitely many degrees of freedom than the often preferred cobordism setting originally proposed by Atiyah \cite{Ati:tqft}. On the other hand it embeds into the larger program on the foundations of quantum theory known as the \emph{general boundary formulation} and provides a direct connection to the \emph{positive formalism} \cite{Oe:dmf,Oe:posfound}. Other TQFT-type approaches to quantizing classical field theory (beyond dimension~2) include the program by Cattaneo, Mnev, Reshetikhin and collaborators based on the BV-BFV formalism \cite{CaMnRe:semiquantft,CaMnRe:pertqgaugebdy} as well as work by Kandel~\cite{Kan:fqftriem}. Works of the former group even include remarks on abelian Yang--Mills theory, mostly on the classical level. However, as far as we are aware, a TQFT-type quantization of abelian Yang--Mills theory in dimensions higher than 2 is achieved in the present work for the first time.

Instead of directly attempting to quantize a classical theory in terms of its geometric and analytic structures such as bundles, sections, spaces of solutions etc.\ we introduce an intermediate step. This step consists of an axiomatization of the classical theory using the same geometric structures as the quantum theory. That is, we associate data to spacetime regions and hypersurfaces and describe its behavior under gluing. The quantization then proceeds in two parts. The first part consists of showing that and how the classical field theory satisfies the classical axioms. The second part consists in constructing the quantum theory from the data of the classical axioms.

This strategy was first proposed and carried through successfully in \cite{Oe:holomorphic}. There, the simplest case of \emph{linear field theory} without gauge symmetries was considered. It was shown how (the second part of) the quantization can be carried out functorially, given that additional data is provided. This data takes the form of a complex structure for a geometric quantization with K\"ahler polarization per hypersurface.

A generalization of the quantization functor to \emph{affine field theory}, i.e., where spaces of solutions are affine spaces, was carried out in \cite{Oe:affine}. This seems in principle a fitting setting for abelian Yang--Mills theory since its spaces of solutions are naturally affine spaces. However, the gauge symmetries of Yang--Mills theory provide a considerable complication. It was shown in~\cite{Dia:gbclassab} how the gauge symmetries can be dealt with by symplectic reduction in such a way that the theory on Riemannian manifolds can be brought into a classical axiomatic formulation similar to the one introduced for affine field theory. In the process, an additional obstruction was dis\-co\-ve\-red. Namely, it turns out that the reduced spaces of solutions in manifolds are not necessarily Lagrangian submanifolds of the reduced boundary data, but merely isotropic ones. In~\cite{Dia:dtonopab} the first step was taken to construct a complex structure that would allow to complete the classical data in order to proceed to (the second part of) the quantization. In the present paper we bring all these ingredients together and show how to complete the picture in order to obtain a~quantization of abelian Yang--Mills theory on Riemannian manifolds with boundary in terms of GBQFT.

In Section~\ref{sec:classax} the axiomatic system for affine classical field theory is introduced. This axioma\-tic data is quantized functorially in Section~\ref{sec:quantization}, leading to GBQFT. (This is step two in the above description.) The construction of the classical data from abelian Yang--Mills theory is presented in Section~\ref{sec:classdem}. (This is step one in the above description.) The special case of abelian quantum Yang--Mills theory on 2-dimensional compact manifolds is worked out more explicitly in Section~\ref{sec:examples}. An outlook is provided in Section~\ref{sec:outlook}. Appendices on the axiomatization of spacetime (Appendix~\ref{sec:geomax}), the core axioms of GBQFT (Appendix~\ref{sec:coreaxioms}) and some basic facts on abelian Yang--Mills theory (Appendix~\ref{sec:minimal-YM}) are also included.

\section{Semiclassical axioms}\label{sec:classax}

For the quantization of abelian Yang--Mills theory on Riemannian manifolds as a general boun\-da\-ry quantum field theory (GBQFT) we can make use of the machinery already developed for linear~\cite{Oe:holomorphic} and affine field theory~\cite{Oe:affine} in this context. Indeed, in those works a quantization functor was exhibited that sends a classical field theory to a general boundary quantum field theory (up to a certain integrability condition). For this purpose the classical field theory is encoded not in terms of fields and differential equations, but rather through algebraic data in axiomatic form. These data involve local solution spaces on a spacetime system. It turns out that abelian Yang--Mills can be brought into a form that ``almost'' satisfies the axioms of affine field theory as given in~\cite{Oe:affine}.

In this section we present a suitably generalized axiomatic system for affine field theory. In Section~\ref{sec:classdem} we then show how abelian Yang--Mills theory gives rise to data that satisfy this system. At that point the motivation and physical meaning of this data will be clarified, see also \cite{Oe:affine}. For the moment we restrict attention to the axiomatic system itself to subsequently consider its quantization in Section~\ref{sec:quantization}. Given a \emph{spacetime system} (reviewed in Appendix~\ref{sec:geomax}), we say that the following axioms determine a \emph{semiclassical affine field theory}.
\begin{itemize}\itemsep=0pt
\item[\textbf{(C1)}] Associated to each hypersurface $\Sigma$ is a complex separable Hilbert space $L_\Sigma$ and an affine space $A_\Sigma$ over $L_\Sigma$ with the induced topology. The latter means that there is a transitive and free abelian group action $L_\Sigma\times A_\Sigma\to A_\Sigma$ which we denote by $(\phi,\eta)\mapsto \phi+\eta$. The inner product in $L_\Sigma$ is denoted by $\{\cdot ,\cdot\}_\Sigma$. We also define $g_\Sigma(\cdot,\cdot)\defeq \Re\{\cdot ,\cdot\}_\Sigma$ and $\omega_\Sigma(\cdot,\cdot)\defeq \frac{1}{2}\Im\{\cdot ,\cdot\}_\Sigma$ and denote by $J_\Sigma\colon L_\Sigma\to L_\Sigma$ the scalar multiplication with $\im$ in $L_\Sigma$. Moreover we suppose there are continuous maps $\theta_\Sigma\colon A_\Sigma\times L_\Sigma\to\R$ and $[\cdot,\cdot]_\Sigma\colon L_\Sigma\times L_\Sigma\to\R$ such that $\theta_\Sigma$ is real linear in the second argument, $[\cdot,\cdot]_\Sigma$ is real bilinear, and both structures are compatible via
\begin{gather*}
[\phi,\phi']_\Sigma+\theta_\Sigma(\eta,\phi')=\theta_\Sigma(\phi+\eta,\phi'),\qquad\forall\, \eta\in A_\Sigma, \qquad \forall\, \phi,\phi'\in L_\Sigma .
\end{gather*}
Finally we require
\begin{gather*}
 \omega_\Sigma(\phi,\phi')=\tfrac{1}{2} [\phi,\phi']_\Sigma-\tfrac{1}{2} [\phi',\phi]_\Sigma, \qquad \forall\, \phi,\phi'\in L_\Sigma .
\end{gather*}
\item[\textbf{(C2)}] Associated to each hypersurface $\Sigma$ there is a homeomorphic involution $A_\Sigma\to A_{\overline{\Sigma}}$ and a~compatible conjugate linear involution $L_\Sigma\to L_{\overline\Sigma}$ under which the inner product is complex conjugated. We will not write these maps explicitly, but rather think of $A_\Sigma$ as identified with $A_{\overline{\Sigma}}$ and $L_\Sigma$ as identified with $L_{\overline\Sigma}$. Then, $\{\phi',\phi\}_{\overline{\Sigma}}=\overline{\{\phi',\phi\}_\Sigma}$ and we also require $\theta_{\overline{\Sigma}}(\eta,\phi)=-\theta_\Sigma(\eta,\phi)$ and $[\phi,\phi']_{\overline{\Sigma}}=-[\phi,\phi']_\Sigma$ for all $\phi,\phi'\in L_\Sigma$ and $\eta\in A_\Sigma$.
\item[\textbf{(C3)}] Suppose the hypersurface $\Sigma$ decomposes into a union of hypersurfaces $\Sigma=\Sigma_1\cup\cdots\cup\Sigma_n$. Then, there is a homeomorphism $A_{\Sigma_1}\times\dots\times A_{\Sigma_n}\to A_\Sigma$ and a compatible isometric isomorphism of complex Hilbert spaces $L_{\Sigma_1}\oplus\cdots\oplus L_{\Sigma_n}\to L_\Sigma$. Moreover, these maps satisfy obvious associativity conditions. We will not write these maps explicitly, but rather think of them as identifications. Also, $\theta_\Sigma=\theta_{\Sigma_1}+\dots+\theta_{\Sigma_n}$ and $[\cdot,\cdot]_\Sigma=[\cdot,\cdot]_{\Sigma_1}+\dots+[\cdot,\cdot]_{\Sigma_n}$.
\item[\textbf{(C4)}] Associated to each region $M$ is a real topological vector space $L_M$ and an affine space $A_M$ over $L_M$ with the induced topology. Also, there is a map $S_M\colon A_M\to\R$.
\item[\textbf{(C5)}] Associated to each region $M$ there is a closed Hilbert subspace $L_{M,\partial M}\subseteq L_{\partial M}$ and an affine subspace $A_{M,\partial M}\subseteq A_{\partial M}$ over it.
\item[\textbf{(C6)}] Associated to each region $M$ there is a continuous map $a_M\colon A_M\to A_{\partial M}$ and a compatible continuous linear map of real vector spaces $r_M\colon L_M\to L_{\partial M}$. We denote by $A_{\tilde{M}}$ the image of $A_M$ under $a_M$ and by $L_{\tilde{M}}$ the image of $L_M$ under $r_M$. Then, $A_{\tilde{M}}\subseteq A_{M,\partial M}$. Also, $L_{\tilde{M}}$ is a closed Lagrangian subspace of the space $L_{M,\partial M}$ as a real symplectic vector space with respect to the symplectic form $\omega_{\partial M}$. We also require
\begin{gather}
S_M(\eta)= S_M(\eta')-\tfrac{1}{2}\theta_{\partial M}(a_M(\eta),r_M(\eta-\eta')) \nonumber\\
\hphantom{S_M(\eta)=}{} -\tfrac{1}{2}\theta_{\partial M}(a_M(\eta'),r_M(\eta-\eta')),\qquad\forall\,\eta,\eta'\in A_M .\label{eq:actsympot}
\end{gather}
\item[\textbf{(C7)}]
Given a hypersurface $\Sigma$ we have for the associated slice region $\hat{\Sigma}$ the equalities $A_{\hat{\Sigma},\partial \hat{\Sigma}}= A_{\partial \hat{\Sigma}}$ and $L_{\hat{\Sigma},\partial \hat{\Sigma}}= L_{\partial \hat{\Sigma}}$. Also, $A_{\hat{\Sigma}}$ can be identified with $A_{\Sigma}$ as a topological affine space and~$L_{\hat{\Sigma}}$ with~$L_{\Sigma}$ as a real topological vector space. Moreover, using these identifications, $a_{\hat{\Sigma}}(\eta)=(\eta,\eta)$ and $r_{\hat{\Sigma}}(\phi)=(\phi,\phi)$ with the decompositions of (C3) understood.
\item[\textbf{(C8)}] Let $M_1$ and $M_2$ be regions and $M\defeq M_1\sqcup M_2$ be their disjoint union. Then, there is a homeomorphism $A_{M_1}\times A_{M_2}\to A_M$ and a compatible isomorphism of real topological vector spaces $L_{M_1}\oplus L_{M_2}\to L_M$ such that $a_M=a_{M_1}\times a_{M_2}$ and $r_M=r_{M_1}\times r_{M_2}$. Moreover, these maps satisfy obvious associativity conditions. Hence, we can think of them as identifications and omit their explicit mention in the following. We also require $S_M=S_{M_1}+S_{M_2}$. Moreover, writing identifications as equalities we require $A_{M,\partial M}=A_{M_1,\partial M_1}\times A_{M_2,\partial M_2}$ and $L_{M,\partial M}=L_{M_1,\partial M_1}\oplus L_{M_2,\partial M_2}$.
\item[\textbf{(C9)}] Let $M$ be a region with its boundary decomposing as a union $\partial M=\Sigma_1\cup\Sigma\cup \overline{\Sigma'}$, where~$\Sigma'$ is a copy of $\Sigma$. Let $M_1$ denote the gluing of~$M$ to itself along $\Sigma$, $\overline{\Sigma'}$ and suppose that~$M_1$ is a~region. Note $\partial M_1=\Sigma_1$. Then, there is an injective map $a_{M;\Sigma,\overline{\Sigma'}}\colon A_{M_1}\toi A_{M}$ and a~compatible injective linear map $r_{M;\Sigma,\overline{\Sigma'}}\colon L_{M_1}\toi L_{M}$ such that
\begin{gather*}
 A_{M_1}\toi A_{M}\rightrightarrows A_\Sigma, \qquad L_{M_1}\toi L_{M}\rightrightarrows L_\Sigma%\label{eq:exsintbdy}
\end{gather*}
are exact sequences. Here, for the first sequence, the arrows on the right hand side are compositions of the map $a_M$ with the projections of $A_{\partial M}$ to $A_\Sigma$ and $A_{\overline{\Sigma'}}$ respectively (the latter identified with $A_\Sigma$). For the second sequence the arrows on the right hand side are compositions of the map $r_M$ with the projections of $L_{\partial M}$ to $L_\Sigma$ and $L_{\overline{\Sigma'}}$ respectively (the latter identified with $L_\Sigma$). We also require $S_{M_1}=S_M\circ a_{M;\Sigma,\overline{\Sigma'}}$.

Consider the projection map $\alpha_1\colon A_{\partial M}\rightarrow A_{\partial M_1}$ and the associated linear projection map $\lambda_1\colon L_{\partial M}\rightarrow L_{\partial M_1}$. Then, $\alpha_1(A_{M,\partial M})\subseteq A_{M_1,\partial M_1}$ and $\lambda_1(L_{M,\partial M})\subseteq L_{M_1,\partial M_1}$ and the following diagrams commute
\begin{gather*}\xymatrix{
 A_{M_1}
 \ar[rrr]^{a_{M;\Sigma\overline{\Sigma'}}}
 \ar[dr] \ar@{-->}[dd]^{{a}_{M_1}}
 &&&
 A_{M}
 \ar[dl] \ar@{-->}[dd]^{{a}_M}
 \\
 &
 A_{M_1,\partial M_1}
 \ar@^{{(}->}[dl]
 &
 A_{M,\partial M}
 \ar@^{{(}->}[dr] \ar[l]_{\alpha_1}
 &
 \\
 A_{\partial M_1}
 &&&
 A_{\partial M}
 \ar[lll]
}\end{gather*}
and
\begin{gather*}\xymatrix{
 L_{M_1}
 \ar[rrr]^{r_{M;\Sigma\overline{\Sigma'}}}
 \ar[dr] \ar@{-->}[dd]^{{r}_{M_1}}
 &&&
 L_{M}
 \ar[dl] \ar@{-->}[dd]^{{r}_M}
 \\
 &
 L_{M_1,\partial M_1}
 \ar@^{{(}->}[dl]
 &
 L_{M,\partial M}
 \ar@^{{(}->}[dr] \ar[l]_{\lambda_1}
 &
 \\
 L_{\partial M_1}
 &&&
 L_{\partial M}.
 \ar[lll]
}\end{gather*}
\end{itemize}

We comment on the difference to the axiomatic system presented in \cite{Oe:affine}. In the latter it is required that for any region $M$ the space $L_{\tilde{M}}$ is a Lagrangian subspace of $L_{\partial M}$. This amounts to requiring $L_{M,\partial M}=L_{\partial M}$ and $A_{M,\partial M}=A_{\partial M}$, resulting in a considerable simplification of the axioms. The more general version of the axioms presented here is motivated precisely by the discovery in \cite{Dia:gbclassab} that for abelian Yang--Mills theory the stricter Lagrangian subspace condition is not satisfied in general. We shall demonstrate in Section~\ref{sec:classdem}, however, that the present generalized axioms are satisfied. Before that, we show in Section~\ref{sec:quantization} that the quantization functor of \cite{Oe:affine} can be generalized correspondingly. This gives us abelian quantum Yang--Mills theory on Riemannian manifolds as a GBQFT.

We recall the following elementary lemmas, adapted to the present setting.

\begin{lem}[\cite{Oe:holomorphic}]For a region $M$ the space $L_{M,\partial M}$ decomposes as a real orthogonal direct sum over $\R$ as $L_{M,\partial M}=L_{\tilde{M}}\oplus J_{\partial M} L_{\tilde {M}}$.
\end{lem}

\begin{lem}[\cite{Oe:affine}]For a region $M$ the space $A_{M,\partial M}$ decomposes as a generalized direct sum over $\R$ as $A_{M,\partial M}=A_{\tilde{M}}\oplus J_{\partial M} L_{\tilde {M}}$.
\end{lem}

We denote in the following the (complex) orthogonal complement of $L_{M,\partial M}$ in $L_{\partial M}$ by $L_{M,\partial M}^\perp$.

\begin{cor}\label{cor:declm}For a region $M$ the space $L_{\partial M}$ decomposes as a direct sum over $\R$ as
\begin{gather}\label{eq:declm}
 L_{\partial M}=L_{\tilde{M}}\oplus J_{\partial M} L_{\tilde {M}}\oplus L_{M,\partial M}^\perp .
\end{gather}
\end{cor}

\begin{cor}\label{cor:decam}For a region $M$ the space $A_{\partial M}$ decomposes as a generalized direct sum over $\R$ as
\begin{gather*}%\label{eq:decam}
 A_{\partial M}=A_{\tilde{M}}\oplus J_{\partial M} L_{\tilde {M}}\oplus L_{M,\partial M}^\perp .
\end{gather*}
\end{cor}

\section{Quantization}\label{sec:quantization}

In the present section we exhibit the quantization functor that assigns to a semiclassical affine field theory in terms of the axioms of Section~\ref{sec:classax} a corresponding GBQFT in terms of the axioms of Appendix~\ref{sec:coreaxioms} in a constructive manner.\footnote{Strictly speaking, the exhibited quantization prescription is not quite a functor since the quantum theory might have restrictions on allowable gluings not present in the classical theory. These come from an integrability condition.} This functor is a generalization of the one given in~\cite{Oe:affine}. Correspondingly, we shall rely heavily on the results presented in \cite{Oe:affine}. These in turn were obtained by recurrence to the linear theory and the functor presented in~\cite{Oe:holomorphic}. As we will also make use of this recurrence here, we start by considering the linear theory.
We shall use superscripts $L$ to distinguish the output of the quantization functor for the linear theory from that of the affine theory, to be considered subsequently.

\subsection{Linear theory}

In this subsection we consider the special case of linear field theory. In terms of the axiomatic system of Section~\ref{sec:classax} this means that the spaces $A_{\Sigma}$ and $A_M$ of local solutions on hypersurfaces and in regions are linear spaces, i.e., have distinguished base points $\varphi_{\Sigma}$ and $\varphi_M$ that fit together under hypersurface decompositions and gluings of regions. The spaces $A_{\Sigma}$ and $A_M$ can then be identified with their vector space counterparts $L_{\Sigma}$ and $L_M$, simplifying considerably the axiomatic system. The resulting system is then a generalization of the axiomatic system for linear field theory presented in \cite[Section~4.1]{Oe:holomorphic}. The latter is recovered completely by always setting $L_{M,\partial M}=L_M$ and eliminating mention of the action from the axioms. In case that the spacetime system arises in terms of submanifolds of a global manifold the choice of base points is essentially equivalent to a choice of global solution $\varphi$ of which $\varphi_{\Sigma}$ and $\varphi_M$ are local restrictions.

The following is largely a review of~\cite[Section~4]{Oe:holomorphic}. We shall indicate where new results are presented. We recall that for a hypersurface $\Sigma$ the real inner product $\frac{1}{2} g_\Sigma$ on the complex Hilbert space $L_\Sigma$ defines a Gaussian measure $\nu_\Sigma$ on the space $\hat{L}_\Sigma$ which is an extension of $L_{\Sigma}$. More precisely, $\hat{L}_\Sigma$ can be identified with the algebraic dual of the topological dual of $L_\Sigma$ so that there is a natural inclusion $L_\Sigma\toi \hat{L}_\Sigma$. Recall also that the square-integrable holomorphic functions on~$\hat{L}_\Sigma$ form a separable complex Hilbert space $\rH^2\big(\hat{L}_\Sigma,\nu_\Sigma\big)$, whose elements are uniquely determined by their values on the subspace $L_\Sigma$. This is declared to be the \emph{state space} $\cHl_\Sigma$ of the quantized theory. This construction is also called the \emph{holomorphic representation} and the elements of $\cHl_\Sigma$ viewed as functions on $L_\Sigma$ or $\hat{L}_\Sigma$ are referred to as \emph{(holomorphic) wave functions}. The conjugate linear isometry $\iota_{\Sigma}^\text{L}\colon \cHl_{\Sigma}\to \cHl_{\overline{\Sigma}}$ associated to orientation change of the hypersurface $\Sigma$ is given by complex conjugation of the wave function. For the decomposition of a hypersurface $\Sigma=\Sigma_1\cup\Sigma_2$ it is clear that we have $\cHl_{\Sigma}=\cHl_{\Sigma_1}\ctens\cHl_{\Sigma_2}$ where $\ctens$ denotes the completed tensor product, since $L_{\Sigma}=L_{\Sigma_1}\oplus L_{\Sigma_2}$. Thus, the quantized theory satisfies Axioms~(T1),~(T1b),~(T2),~(T2b).

A particularly important set of states are the \emph{coherent states}. These are the usual coherent states in linear field theory that can be obtained as exponentials of creation operators. They are parametrized by elements of $L_{\Sigma}$. The coherent state $K_{\xi}\in \cHl_\Sigma$ associated to the local solution $\xi\in L_{\Sigma}$ is given by the wave function
\begin{gather}
 K_{\xi}(\phi)=\exp\big(\tfrac{1}{2}\{\xi,\phi\}_{\Sigma}\big),\qquad\forall\,\phi\in L_{\Sigma} . \label{eq:lcoh}
\end{gather}
Key properties of these states are the \emph{reproducing property} and the \emph{completeness property},
\begin{gather*}
 \langle K_{\xi},\psi\rangle_{\Sigma}^{\text{L}}=\psi(\xi), \qquad\forall\,\xi\in L_{\Sigma},\qquad\forall\, \psi\in\cHl_{\Sigma},\\
 \langle \psi',\psi\rangle_{\Sigma}^{\text{L}}=
 \int_{\hat{L}_{\Sigma}} \langle\psi',K_{\xi}\rangle_{\Sigma}^{\text{L}} \langle K_{\xi},\psi\rangle_{\Sigma}^{\text{L}}\, \xd\nu_{\Sigma}(\xi),\qquad\forall\, \psi,\psi'\in\cHl_{\Sigma} .
\end{gather*}
Further properties of coherent states are
\begin{gather*}
 \langle K_{\xi'}, K_{\xi}\rangle_{\Sigma}^{\text{L}}=\exp\left(\tfrac{1}{2}\{\xi,\xi'\}_{\Sigma}\right), \qquad\forall\, \xi,\xi'\in L_{\Sigma},\\
 \iota_{\Sigma}^{\text{L}}(K_{\Sigma,\xi})=K_{\overline{\Sigma},\xi}, \qquad\forall\,\xi\in L_{\Sigma},\\
 K_{\Sigma_1\cup\Sigma_2, (\xi_1,\xi_2)}=K_{\Sigma_1,\xi_1}\tens K_{\Sigma_2,\xi_2}, \qquad\forall\,\xi_1\in L_{\Sigma_1}, \xi_2\in L_{\Sigma_2} .
\end{gather*}
We denote the \emph{normalized} version of the coherent state $K_\xi$ by $\tilde{K}_\xi$.

Given a spacetime region $M$ the inner product $\frac{1}{4} g_{\partial M}$ restricted to $L_{\tilde{M}}$ defines a Gaussian measure on the space $\hat{L}_{\tilde{M}}$ that we denote by $\nu_M$. Here, $\hat{L}_{\tilde{M}}$ can be identified with the algebraic dual of the topological dual of~$L_{\tilde{M}}$. Note that the measure $\nu_M$ (denoted~$\nu_{\tilde{M}}$ in~\cite{Oe:affine, Oe:holomorphic}) is distinct from the measure $\nu_{\partial M}$ restricted to the same space. The amplitude map $\rhol_M\colon \cHld_{\partial M}\to\C$ is given by the integral (the map~$r_M$ being implicit in our notation)
\begin{gather}
 \rhol_M(\psi)=\int_{\hat{L}_{\tilde{M}}} \psi(\phi)\,\xd\nu_M(\phi) .
 \label{eq:linampl}
\end{gather}
Note that $\psi$ is square integrable with respect to $\nu_{\partial M}$. This does not mean that it is integrable with respect to $\nu_M$. Indeed, we shall say that $\rhol_M$ is defined for $\psi$ precisely if this is the case. We denote the subspace of $\cHl_{\partial M}$ with this property by $\cHld_{\partial M}$. As we shall see this includes at least all coherent states and their linear combinations. As these form a dense subspace of~$\cHl_{\partial M}$,~$\cHld_{\partial M}$ is dense. It was shown in~\cite{Oe:feynobs} that the prescription~(\ref{eq:linampl}) is equivalent to Feynman path integral quantization.

The following result gives the value of the amplitude map on a coherent state. This generalizes Proposition~4.2 of \cite{Oe:holomorphic}.
\begin{prop} \label{prop:lamplcoh} Let $M$ be a region and $\xi\in L_{\partial M}$. Write $\xi=\xi^{\text{R}}+J_{\partial M}\xi^{\text{I}}+\xi^0$ in terms of the decomposition \eqref{eq:declm}. Then, $K_{\xi}\in\cHld_{\partial M}$ and
 \begin{gather*}
 \rhol_M(K_{\xi})=\exp\big(\tfrac{1}{4}g_{\partial M}\big(\xi^{\text{R}},\xi^{\text{R}}\big)
 -\tfrac{1}{4}g_{\partial M}\big(\xi^{\text{I}},\xi^{\text{I}}\big)-\tfrac{\im}{2}g_{\partial M}\big(\xi^{\text{R}},\xi^{\text{I}}\big)\big) .
 \end{gather*}
\end{prop}
\begin{proof} Observe that $\xi^0$ is complex orthogonal to $L_{\tilde M}$. That is, given $\phi\in L_{\tilde{M}}$ we have $K_{\xi}(\phi)=K_{\xi-\xi^0}(\phi)$. By inspection of~(\ref{eq:linampl}) we see that we must have $\rhol_M(K_{\xi})=\rhol_M(K_{\xi-\xi^0})$ if the amplitude is defined. The statement reduces then to that of Proposition~4.2 of~\cite{Oe:holomorphic}.
\end{proof}

With this we satisfy Axiom~(T4). Also, Axiom (T3x) is satisfied. The proof reduces to that given in~\cite{Oe:holomorphic} due to the fact that $L_{\partial\hat{\Sigma}}=L_{\hat{\Sigma},\partial\hat{\Sigma}}$ for slice regions $\hat{\Sigma}$, see Axiom~(C7). Axiom~(T5a) is also immediate.
\begin{cor} For the corresponding normalized coherent state $\tilde{K}_\xi$ we get
 \begin{gather}
 \rhol_M(\tilde{K}_{\xi})=\exp\big({-}\tfrac{\im}{2}g_{\partial M}\big(\xi^{\text{R}},\xi^{\text{I}}\big)
 -\tfrac{1}{2}g_{\partial M}\big(\xi^{\text{I}},\xi^{\text{I}}\big)
 -\tfrac{1}{4}g_{\partial M}\big(\xi^{0},\xi^{0}\big)\big) . \label{eq:lncohampl}
 \end{gather}
\end{cor}

Recall that a simple, but compelling physical interpretation of the amplitude formula~(\ref{eq:lncohampl}) was put forward in~\cite{Oe:holomorphic}, valid here in the special case $L_{M,\partial M}=L_{\partial M}$. Essentially this same interpretation extends to the present more general setting. If we think in classical terms, the component $\xi^{\mathrm{R}}$ of the boundary solution $\xi$ can be continued consistently to the interior and is hence classically allowed. The components $J_{\partial M}\xi^{\mathrm{I}}$ and $\xi^0$ do not possess such a continuation and are hence classically forbidden. This is reflected precisely in equation~(\ref{eq:lncohampl}). If the classically forbidden components are not present, the amplitude has unit value. On the other hand, the presence of a classically forbidden component leads to an exponential suppression, governed precisely by the ``magnitude'' of this component (measured in terms of the metric $g_{\partial M}$). It is interesting to note that the suppression factor is not the same for $J_{\partial M}\xi^{\mathrm{I}}$ and $\xi^0$. However, it is not clear whether this difference can be given a simple physical interpretation.

It remains to show that the quantized theory satisfies the gluing Axiom (T5b). Consider a~region~$M$ with boundary decomposition $\partial M=\partial M_1\cup\Sigma\cup\overline{\Sigma'}$ gluable to itself along~$\Sigma$ with~$\Sigma'$ resulting in the region $M_1$. Using the completeness relation of the coherent states the gluing identity~(\ref{eq:glueid}) can be rewritten as
\begin{gather}
 \rho_{M_1}^{\mathrm{L}}(\psi)\cdot c\big(M;\Sigma,\overline{\Sigma'}\big)=\int_{\hat{L}_\Sigma}\rho_M^{\mathrm{L}}(\psi\tens K_{\xi}\tens\iota_\Sigma(K_{\xi}))\, \xd\nu_\Sigma(\xi) ,\label{eq:lglueid}
\end{gather}
for any $\psi\in\cHld_{\partial M_1}$.
Recall that for Axiom (T5b) to hold we must require an additional \emph{integrability condition}. We say that the gluing data satisfies the integrability condition if the function $L_{\Sigma}\to\C$ given by
\begin{gather}
 \xi\mapsto \rhol_M(K_0\tens K_{\xi}\tens\iota_{\Sigma}(K_{\xi}))
 \label{eq:lint}
\end{gather}
extends to an integrable function on $\hat{L}_{\Sigma}$ with respect to the measure $\nu_{\Sigma}$. The validity of \mbox{Axiom~(T5b)} is then the subject of the following result, generalizing Theorem~4.5 of \cite{Oe:holomorphic}.

\begin{thm} \label{thm:lgluing}If the integrability condition is satisfied, then Axiom~{\rm (T5b)} holds. Moreover, the \emph{gluing anomaly factor} is given by
\begin{gather}
 c(M;\Sigma,\overline{\Sigma'}) = \int_{\hat{L}_{\Sigma}} \rhol_M(K_0\tens K_{\xi}\tens \iota_{\Sigma}(K_{\xi}))\,\xd\nu_{\Sigma}(\xi) . \label{eq:lanom}
\end{gather}
\end{thm}
\begin{proof} It is sufficient to show the validity of the gluing identity (\ref{eq:lglueid}) for coherent states. That is, we need to show, for any $\phi\in L_{\partial M_1}$,
 \begin{gather}
 \rhol_{M_1}(K_{\phi})\cdot c\big(M;\Sigma,\overline{\Sigma'}\big) = \int_{\hat{L}_{\Sigma}} \rhol_M(K_{\phi}\tens K_{\xi}\tens \iota_{\Sigma}(K_{\xi}))\,\xd\nu_{\Sigma}(\xi) . \label{eq:gluingidcoh}
 \end{gather}
 We decompose $\phi=\phi^{\text{X}}+\phi^0$ with $\phi^{\text{X}}\in L_{M_1,\partial M_1}$ and $\phi^0\in L_{M_1,\partial M_1}^\perp$. From Proposition~\ref{prop:lamplcoh} we have on the left-hand side
 $\rhol_{M_1}(K_{\phi})=\rhol_{M_1}(K_{\phi^{\text{X}}})$. On the other hand, the integrand on the right-hand side can be rewritten as
 \begin{gather*}
 \rhol_M(K_{\phi}\tens K_{\xi}\tens \iota_{\Sigma}(K_{\xi})) =\int_{\hat{L}_{\tilde{M}}} K_{\phi}(\eta_1) K_\xi(\eta_{\Sigma}) \overline{K_\xi(\eta_{\Sigma'})}\,\xd\nu_M(\eta) .
 \end{gather*}
 Here, $(\eta_1,\eta_{\Sigma},\eta_{\Sigma'})=r_M(\eta)$. However, since $\lambda_1(L_{M,\partial M})\subseteq L_{M_1,\partial M_1}$ with Axiom (C9) we have $\eta_1\in L_{M_1,\partial M_1}$ and thus $\{\phi,\eta_1\}_{\partial M_1}^{\text{L}}=\{\phi^{\text{X}},\eta_1\}_{\partial M_1}^{\text{L}}$. From the explicit form of the wave func\-tion~(\ref{eq:lcoh}) of the coherent states this implies $K_{\phi}(\eta_1)=K_{\phi^{\text{X}}}(\eta_1)$. Thus on the right-hand side of~(\ref{eq:gluingidcoh}) we can also replace $\phi$ by $\phi^{\text{X}}$. The proof then reduces to the proof of~Theorem~4.5 of~\cite{Oe:holomorphic} by replacing $L_{\partial M}$ and $L_{\partial M_1}$ there with $L_{M,\partial M}$ and $L_{M_1,\partial M_1}$ respectively.
\end{proof}

\subsection{Affine theory}

We proceed in the present subsection to describe the functor that constructs from an affine field theory in terms of the axioms of Section~\ref{sec:classax} a GBQFT in terms of the axioms of Appendix~\ref{sec:coreaxioms}. As this functor is an adaption of the one presented in~\cite{Oe:affine}, large parts of this subsection consist of a review of Section~4 of that paper.

Given a hypersurface $\Sigma$ we denote the algebra of complex valued continuous functions on $A_\Sigma$ by $\cont_\Sigma$. We define the Hilbert space $\cH_\Sigma$ associated to the hypersurface $\Sigma$ as a certain subspace of $\cont_\Sigma$ as follows. Fix a \emph{base point} $\eta\in A_\Sigma$ and consider the following element of $\cont_\Sigma$
\begin{gather*}
\alpha_{\Sigma}^\eta(\varphi)=\exp\big(\tfrac{\im}{2} \theta_\Sigma(\eta,\varphi-\eta) + \tfrac{\im}{2} \theta_\Sigma(\varphi,\varphi-\eta)-\tfrac{1}{4} g_\Sigma(\varphi-\eta,\varphi-\eta)\big) .
\end{gather*}
Define $\cH_\Sigma$ as the subspace of $\cont_\Sigma$ of elements $\psi$ that take the form
\begin{gather}
 \psi(\varphi)=\chi^\eta(\varphi-\eta)\alpha_{\Sigma}^\eta(\varphi) , \label{eq:wfdec}
\end{gather}
where $\chi^\eta\in\cHl_{\Sigma}$. Equip $\cH_\Sigma$ with the inner product
\begin{gather*}
\langle\psi',\psi\rangle_\Sigma = \int_{\hat{L}_\Sigma} \chi^\eta\overline{{\chi'}^\eta} \xd\nu_\Sigma .
\end{gather*}
$\cH_\Sigma$ becomes a complex separable Hilbert space in this way, naturally isomorphic to $\cHl_\Sigma$. Crucially, the definition is independent of the choice of base point as shown by Lemma~4.1 of \cite{Oe:affine}.

As for the elements of $\cHl_{\Sigma}$ we shall refer to the elements of $\cH_\Sigma$ also as \emph{wave functions}. Note that $\alpha_{\overline{\Sigma}}^{\eta}=\overline{\alpha_{\Sigma}^{\eta}}$. Thus, complex conjugation of wave functions yields a conjugate linear isomorphism $\iota_\Sigma\colon \cH_\Sigma\to\cH_{\overline{\Sigma}}$. For a hypersurface decomposition $\Sigma=\Sigma_1\cup\Sigma_2$ and $\eta_1\in A_{\Sigma_1}$, $\eta_2\in A_{\Sigma_2}$ we have $\alpha_{\Sigma_1\cup\Sigma_2}^{(\eta_1,\eta_2)}=\alpha_{\Sigma_1}^{\eta_1}\alpha_{\Sigma_2}^{\eta_2}$ and therefore naturally get $\cH_{\Sigma_1\cup\Sigma_2}=\cH_{\Sigma_1}\ctens\cH_{\Sigma_2}$, where the tensor product is the completed tensor product of Hilbert spaces. Thus, we satisfy Axioms~(T1), (T1b), (T2), (T2b).

It turns out that in affine field theory there is a natural notion of \emph{affine coherent states}, somewhat, but not completely analogous to the coherent states for linear field theory. For a hypersurface $\Sigma$ the affine coherent states are parametrized by~$A_{\Sigma}$. Given $\zeta\in A_{\Sigma}$ the associated coherent state $\hat{K}_\zeta\in\cH_\Sigma$ is determined by the wave function
\begin{gather*}
\hat{K}_\zeta(\varphi)\defeq\exp\big(\tfrac{\im}{2}\theta_\Sigma(\zeta,\varphi-\zeta)+\tfrac{\im}{2}\theta_\Sigma(\varphi,\varphi-\zeta)-\tfrac{1}{4}g_\Sigma(\varphi-\zeta,\varphi-\zeta)\big) .%\label{eq:acoh}
\end{gather*}
The affine coherent states satisfy the \emph{reproducing property} and a \emph{completeness relation}
\begin{gather}
\langle \hat{K}_\zeta,\psi\rangle_\Sigma=\psi(\zeta),\nonumber\\
 \langle\psi',\psi\rangle_\Sigma=\int_{\hat{L}_\Sigma} \langle\psi',\hat{K}_{\eta+\xi}\rangle_\Sigma
 \langle \hat{K}_{\eta+\xi},\psi\rangle_\Sigma\exp\big(\tfrac{1}{2}g_\Sigma(\xi,\xi)\big) \xd\nu_\Sigma(\xi).\label{eq:acohcompl}
\end{gather}
Note that in the completeness relation we use a base point $\eta\in A_\Sigma$ as we have a measure on $\hat{L}_\Sigma$, rather than on $\hat{A}_\Sigma$. However, the choice of base point is arbitrary.
We also exhibit the inner product
\begin{gather*}
 \langle \hat{K}_{\zeta'},\hat{K}_{\zeta}\rangle_\Sigma=\exp\big(\tfrac{\im}{2}\theta_\Sigma(\zeta,\zeta'-\zeta)+\tfrac{\im}{2}\theta_\Sigma(\zeta',\zeta'-\zeta)-\tfrac{1}{4}g_\Sigma(\zeta'-\zeta,\zeta'-\zeta)\big) .
\end{gather*}
Note in particular, that the affine coherent states are already \emph{normalized}. The behavior of the affine coherent states with respect to orientation change and hypersurface decomposition is straightforward
\begin{gather*}
\hat{K}_{\overline{\Sigma},\zeta}=\iota_{\Sigma}\big(\hat{K}_{\Sigma,\zeta}\big), \qquad
 \hat{K}_{\Sigma\cup\Sigma',(\zeta,\zeta')}=\hat{K}_{\Sigma,\zeta}\tens \hat{K}_{\Sigma',\zeta'} .
\end{gather*}

It is also useful to explicitly relate the affine coherent states to the coherent states of the linear theory. Given a base point $\eta\in A_\Sigma$ and an element $\xi\in L_{\Sigma}$ we consider the state $K_{\xi}^\eta\in\cH_{\Sigma}$ given by the wave function
\begin{gather}
K_{\xi}^\eta(\varphi)\defeq K_\xi(\varphi-\eta)\alpha_{\Sigma}^{\eta}(\varphi),\qquad\forall\, \varphi\in A_\Sigma .\label{eq:deccohwf}
\end{gather}
We observe that this state is related to the affine coherent state $\hat{K}_{\eta+\xi}$ by a constant factor
\begin{gather}
 \hat{K}_{\eta+\xi}=K^\eta_\xi \exp\big({-}\im\theta_\Sigma(\eta,\xi)-\tfrac{\im}{2}[\xi,\xi]_\Sigma-\tfrac{1}{4}g_\Sigma(\xi,\xi)\big) .\label{eq:relcohst}
\end{gather}

Let $M$ be a region. The amplitude map $\rho_M\colon \cH_{\partial M}^{\ds}\to\C$ in $M$ is defined with the help of a~choice of \emph{base point} $\eta\in A_{M}$ as follows. We decompose the wave function of a state $\psi\in\cH_{\partial M}$ as in equation~(\ref{eq:wfdec}). This specifies an element $\chi^{\eta}$ in $\cH_{\partial M}$. In particular, $\chi^{\eta}$ determines a complex function on $\hat{L}_{\partial M}$. If this function is integrable with respect to the measure $\nu_M$ on $\hat{L}_{\tilde{M}}$, then we declare $\psi\in\cHd_{\partial M}$ and the amplitude of $\psi$ to be given by (with $r_M$ implicit)
\begin{gather}
 \rho_M(\psi)=\exp\left(\im S_M(\eta)\right)\int_{\hat{L}_{\tilde{M}}} \chi^\eta(\phi)\,\xd\nu_{M}(\phi) .\label{eq:affampl}
\end{gather}
By Lemma~4.2 of~\cite{Oe:affine} this definition is independent of the base point. It was shown in \cite{Oe:feynobs} that this definition precisely implements Feynman path integral quantization.
As we shall see in a~moment the subspace $\cH^\ds_{\partial M}$ of $\cH_{\partial M}$ of wave functions which are integrable contains at least all coherent states. Thus, $\cH^\ds_{\partial M}$ is dense in $\cH_{\partial M}$ and Axiom (T4) is satisfied. Also, Proposition~4.4 of~\cite{Oe:affine} shows that Axiom~(T3x) is satisfied. The proposition is applicable since $L_{\partial\hat{\Sigma}}=L_{\hat{\Sigma},\partial\hat{\Sigma}}$ and $A_{\partial\hat{\Sigma}}=A_{\hat{\Sigma},\partial\hat{\Sigma}}$ by Axiom~(C7). Axiom~(T5a) is straightforward using the additivity of~$S$ as exhibited in Axiom~(C8).

The following result yields an explicit formula for the amplitude of an affine coherent state. This generalizes Proposition~4.3 of \cite{Oe:affine}.
\begin{prop}\label{prop:cohampl}Let $\zeta\in A_{\partial M}$ and $\zeta=\zeta^{\mathrm{R}}+J_{\partial M} \zeta^{\mathrm{I}}+\zeta^0$ be its decomposition with respect to the generalized direct sum $A_{\partial M}=A_{\tilde{M}}\oplus J_{\partial M} L_{\tilde{M}}\oplus L_{M,\partial M}^\perp$ according to Corollary~{\rm \ref{cor:decam}}. Then
\begin{gather}
 \rho_M\big(\hat{K}_\zeta\big)= \exp\big(\im S_M\big(\zeta^{\mathrm{R}}\big)-\im \theta_{\partial M}\big(\zeta^{\mathrm{R}},J_{\partial M}\zeta^{\mathrm{I}}+\zeta^0\big) \nonumber\\
\hphantom{\rho_M\big(\hat{K}_\zeta\big)=}{} -\tfrac{\im}{2}[J_{\partial M}\zeta^{\mathrm{I}}+\zeta^0,J_{\partial M}\zeta^{\mathrm{I}}+\zeta^0]_{\partial M}-\tfrac{1}{2} g_{\partial M} \big(\zeta^{\mathrm{I}},\zeta^{\mathrm{I}}\big)-\tfrac{1}{4} g_{\partial M}\big(\zeta^0,\zeta^0\big)\big) .\label{eq:cohampl}
\end{gather}
\end{prop}
\begin{proof}Taking $\zeta^{\mathrm{R}}$ as base point and defining $\zeta^{\mathrm{X}}\defeq \zeta-\zeta^0$ we decompose the coherent state wave function as in (\ref{eq:wfdec})
\begin{gather*}
\hat{K}_\zeta(\varphi)=\chi^{\zeta^{\text{R}}}(\varphi-\eta)\alpha_{\partial M}^{\zeta^{\text{R}}}(\varphi) .
\end{gather*}
Combining (\ref{eq:relcohst}) and (\ref{eq:deccohwf}) then yields
\begin{gather*}
\chi^{\zeta^{\text{R}}}=K_{J_{\partial M} \zeta^{\mathrm{I}}+\zeta^0}\exp\big({-}\im \theta_{\partial M}\big(\zeta^{\text{R}},J_{\partial M} \zeta^{\mathrm{I}}+\zeta^0\big) \\
\hphantom{\chi^{\zeta^{\text{R}}}=}{}
-\tfrac{\im}{2}\big[J_{\partial M} \zeta^{\mathrm{I}}+\zeta^0,J_{\partial M} \zeta^{\mathrm{I}}+\zeta^0\big]_{\partial{M}}
-\tfrac{1}{4}g_{\partial{M}}\big(J_{\partial M} \zeta^{\mathrm{I}}+\zeta^0,J_{\partial M} \zeta^{\mathrm{I}}+\zeta^0\big)\big) .
\end{gather*}
The relevant integral in (\ref{eq:affampl}) is
\begin{gather*}
\int_{\hat{L}_{\tilde{M}}} K_{J_{\partial M} \zeta^{\mathrm{I}}+\zeta^0}(\phi)\,\xd\nu_{M}(\phi)
=\int_{\hat{L}_{\tilde{M}}} K_{J_{\partial M} \zeta^{\mathrm{I}}}(\phi)\,\xd\nu_{M}(\phi)
=\exp\big({-}\tfrac{1}{4}g_{\partial M}\big(\zeta^{\mathrm{I}},\zeta^{\mathrm{I}}\big)\big).
\end{gather*}
The first equality comes from the fact that $\zeta^0$ is complex orthogonal to $L_{\tilde{M}}$ and thus does not contribute to the coherent state wave function (\ref{eq:lcoh}) evaluated on $\phi$. The second equality follows from Proposition~4.2 of \cite{Oe:holomorphic}. Combining, formula (\ref{eq:affampl}) yields equation (\ref{eq:cohampl}).
\end{proof}

It remains to demonstrate the validity of the gluing Axiom (T5b). Thus, consider a~region~$M$ with its boundary decomposing as $\partial M=\Sigma_1\cup\Sigma\cup \overline{\Sigma'}$, where $\Sigma'$ is a copy of~$\Sigma$. $M_1$ denotes the gluing of~$M$ to itself along $\Sigma$, $\overline{\Sigma'}$ and we suppose that~$M_1$ is an admissible region. Fixing a~base point $\eta\in A_\Sigma$ the gluing identity~(\ref{eq:glueid}) of Axiom~(T5b) in terms of the affine coherent state completeness relation~(\ref{eq:acohcompl}) takes the form
\begin{gather}
 \rho_{M_1}(\psi)\cdot c\big(M;\Sigma,\overline{\Sigma'}\big)=\int_{\hat{L}_\Sigma}\rho_M(\psi\tens \hat{K}_{\eta+\xi}\tens\iota_\Sigma(\hat{K}_{\eta+\xi}))
\exp\left(\tfrac{1}{2}g_\Sigma(\xi,\xi)\right)\xd\nu_\Sigma(\xi) ,\label{eq:aglueid}
\end{gather}
for all $\psi\in\cHd_{\Sigma_1}$.

As in the linear case, we require an \emph{integrability condition} for the axiom to hold. Here, this condition is that for some, hence any, $\eta\in A_{M_1}$, the extension of the function $L_\Sigma\to\C$ defined by
\begin{gather}
 \xi\mapsto \rho_M\big(\hat{K}_{\eta_1}\tens \hat{K}_{\eta_{\Sigma}+\xi}\tens\iota_\Sigma\big(\hat{K}_{\eta_{\Sigma}+\xi}\big)\big)\exp\left(\tfrac{1}{2}g_\Sigma(\xi,\xi)\right)\label{eq:aint}
\end{gather}
to a function on $\hat{L}_\Sigma$ is $\nu_\Sigma$-integrable and its integral is different from zero. Here $(\eta_1,\eta_{\Sigma},\eta_{\Sigma'})=a_M\circ a_{M;\Sigma,\overline{\Sigma'}}(\eta)$, compare Axioms (C6) and (C9).

\begin{thm}\label{thm:agluing}If the integrability condition is satisfied, then Axiom~{\rm (T5b)} holds. Moreover, given a base point $\eta\in A_{M_1}$
\begin{gather}
c\big(M;\Sigma,\overline{\Sigma'}\big)=\exp (-\im S_{M_1}(\eta) t)\nonumber\\
\hphantom{c\big(M;\Sigma,\overline{\Sigma'}\big)=}{}\times \int_{\hat{L}_\Sigma} \rho_M\big(\hat{K}_{\eta_1}\tens \hat{K}_{\eta_{\Sigma}+\xi}\tens\iota_\Sigma\big(\hat{K}_{\eta_{\Sigma}+\xi}\big)\big)\exp\left(\tfrac{1}{2}g_\Sigma(\xi,\xi)\right) \xd\nu_\Sigma(\xi) ,\label{eq:afaf}
\end{gather}
and this expression is independent of the choice of base point.
\end{thm}
\begin{proof}We perform the proof by reducing it to the corresponding Theorem~\ref{thm:lgluing} of the linear theory. We first relate the amplitude maps of the affine setting explicitly with those of the linear setting. Decomposing the wave function of a state $\psi\in\cH_{\partial N}$ on the boundary of a region $N$ with respect to a base point $\eta\in A_N$ as in~(\ref{eq:wfdec}) we have
\begin{gather*}
\rho_N(\psi)=\exp\left(\im S_N(\eta)\right)\rho_N^{\mathrm{L}}(\chi^\eta) ,%\label{eq:relampl}
\end{gather*}
as seen by inspection of equations (\ref{eq:linampl}) and (\ref{eq:affampl}). In particular, for coherent states of the type exhibited in (\ref{eq:deccohwf}) we obtain with $\xi\in L_{\partial N}$
\begin{gather*}
\rho_N\big(K^\eta_\xi\big)=\exp\left(\im S_N(\eta)\right) \rho_N^{\mathrm{L}}(K_\xi) .%\label{eq:relamplcoh}
\end{gather*}
In the context of the gluing data, upon choosing a base point $\eta\in A_{M_1}$ and using equation~(\ref{eq:relcohst}) we can rewrite the integrand of~(\ref{eq:afaf}) in terms of the linear amplitude map
\begin{gather*}
 \exp\big({-}\im S_M(\eta)+\tfrac{1}{2}g_\Sigma(\xi,\xi)\big)
 \rho_M\big(\hat{K}_{\eta_1}\tens\hat{K}_{\eta_{\Sigma}+\xi}\tens\iota_\Sigma\big(\hat{K}_{\eta_{\Sigma}+\xi}\big)\big) = \rho_M^{\mathrm{L}}(K_0\tens K_\xi\tens\iota_\Sigma(K_\xi)) . %\label{eq:relsampl}
\end{gather*}
This shows that the claimed anomaly factor (\ref{eq:afaf}) is in fact precisely the anomaly factor (\ref{eq:lanom}) of the linear theory. This also shows that the integrability condition~(\ref{eq:aint}) for the affine theory is equivalent to the one~(\ref{eq:lint}) of the linear theory.

We proceed to reduce the gluing identity (\ref{eq:aglueid}) to that of the linear theory (\ref{eq:lglueid}). It will be convenient to use coherent states of the form~(\ref{eq:deccohwf}). We use a base point $\zeta\in A_{M_1}$, but $\phi\in L_{\partial {M_1}}$ arbitrary. With Theorem~\ref{thm:lgluing} we obtain, as desired,
\begin{gather*}
 \rho_{M_1}\big(K^{\zeta_1}_\phi\big)\cdot c(M;\Sigma,\overline{\Sigma'}) =\exp (\im S_{M_1}(\zeta) ) \rho_{M_1}^{\mathrm{L}}(K_\phi)\cdot c\big(M;\Sigma,\overline{\Sigma'}\big)\\
 \qquad{} =\exp\left(\im S_{M_1}(\zeta)\right) \int_{\hat{L}_\Sigma}\rho_M^{\mathrm{L}}(K_\phi\tens K_\xi\tens\iota_\Sigma(K_\xi))\, \xd\nu_\Sigma(\xi)\\
 \qquad{} =\int_{\hat{L}_\Sigma}\rho_M\big(K_\phi^{\zeta_1}\tens K_\xi^{\zeta_{\Sigma}}\tens\iota_\Sigma\big(K_\xi^{\zeta_{\Sigma}}\big)\big)\, \xd\nu_\Sigma(\xi)\\
 \qquad{} =\int_{\hat{L}_\Sigma}\rho_M\big(K_\phi^{\zeta_1}\tens \hat{K}_{\zeta_{\Sigma}+\xi}\tens\iota_\Sigma\big(\hat{K}_{\zeta_{\Sigma}+\xi}\big)\big)\exp\left(\tfrac{1}{2}g_\Sigma(\xi,\xi)\right) \xd\nu_\Sigma(\xi) .\tag*{\qed}
\end{gather*}\renewcommand{\qed}{}
\end{proof}

\section{Semiclassical abelian Yang--Mills fields}\label{sec:classdem}

In the present section we show in a constructive manner how abelian Yang--Mills theory gives rise to the data of semiclassical affine field theory in terms of the axiomatic system of Section~\ref{sec:classax}. Together with the quantization functor of Section~\ref{sec:quantization} this provides quantized abelian Yang--Mills theory as a GBQFT. We work on smooth manifolds equipped with a Riemannian metric. For the reader's convenience we review some of the geometric facts for YM fields in Appendix \ref{sec:minimal-YM}, see also \cite{Dia:gbclassab}. We rely on the results obtained in \cite{Dia:dtonopab}. We start with an overview and motivation of the construction in the more general context of affine field theory.

For the Yang--Mills model, manifolds carry additional structure arising from the metric and from the principal bundle over it. Hypersurfaces should be considered with their metric germs in the ambient manifold. During the gluing procedure we need to provide additional data that encodes the topology of the final principal bundle, see Appendix~\ref{sec:geomax}. In our case a cohomology class $c({\cal E})\in H^2_{\rm dR}(M,\mathbb{Z})$ associated to the final principal bundle carries this additional information.

\subsection{Overview}

Our starting point is affine classical field theory for an action $S_M$ on a spacetime region mo\-deled as a smooth oriented $n$-manifold $M$ with smooth boundary $\partial M\neq \varnothing$. We consider the compact abelian structure Lie group~${\rm U}(1)$. Moreover, for simplicity we fix a principal bundle~${\cal E}$, with fiber~${\rm U}(1)$. Each smooth connection $\eta\in \operatorname{Conn}({\cal E})$ has a closed curvature $2$-form on~$M$, \smash{$F^\eta\in \Omega^2(M)$}. For any other smooth connection~$\eta'$, $\eta-\eta'$ is in correspondence with certain $\gf\in \Omega^1(M)$, hence $F^{\eta}=F^{\eta'}+\xd\gf$, and the cohomology class
\begin{gather*}
 c({\cal E}):=\left[\tfrac{1}{2\pi}F^\eta\right]\in H^2_{\rm dR}(M,\mathbb{Z})
\end{gather*}
does not depend on $\eta\in \operatorname{Conn}({\cal E})$. For other (non-compact) abelian Lie groups without compact factor, there is no need to consider $c({\cal E})$. Thus the choice of ${\rm U}(1)$ forces considerations of the topology of the region $M$ and the bundle ${\cal E}$ in our calculations (expectation values).

We consider the Yang--Mills action
\begin{gather}\label{eqn:YM-action}
 S_M(\eta) = \tfrac{1}{2}\int_MF^{\eta}\wedge \star F^{\eta},\qquad \forall\, \eta\in\operatorname{Conn}({\cal E}),
\end{gather}
where $\star$ stands for the Hodge star operator in $M$.

We suppose that the Euler--Lagrange equations yield an affine space of solutions, $A_M$. For first-order Lagrangian densities, the space of first order boundary data, $A_{\partial M}$, is also an affine space describing both Dirichlet as well as Neumann boundary conditions. Among boundary data, those that can be extended as \emph{fields} in the bulk describe an affine subspace $A_{M,\partial M}\subseteq A_{\partial M}$, while we denote those that can be extended as Euler--Lagrange {\em solutions} in the bulk as $A_{\tilde{M}}\subseteq A_{M,\partial M}$. Since the general boundary value problem is uniquely solved, $A_{\partial M}$ may be regarded as the space of solutions on a small cylinder $\partial M\times[0,\varepsilon)$, while the particular {\em topological restrictions} imposed by the inclusion $\partial M\subseteq M$ may manifest as a proper inclusion $A_{M,\partial M}\subsetneq A_{\partial M}$. The linear spaces associated to the previous affine ones are denoted as $L_M$, $L_{\partial M}$, $L_{\tilde{M}}$ and $L_{M,\partial M}$, respectively.

Geometric quantization requires the presence of a linear presymplectic structure $\widetilde{\omega}_\Sigma$ defined in the Lagrangian setting, see \cite{Woo:geomquant}. According to Appendix~\ref{sec:geomax} for a spacetime system we consider $n$-dimensional spacetime regions $M$ as well as {\em thickened} $(n-1)$-dimensional hypersurfaces $\Sigma$. Thick in this situation means a suitable normal structure attached to $\Sigma$ when we regard it as a Riemannian submanifold. To the pair $(M,\partial M)$ we will associate another pair $(A_{\tilde{M}}, A_{M,\partial M})$ which would define a {\em Lagrangian relation}. We avoid regarding $A_{\tilde{M}}\subseteq A_{M,\partial M}$ as morphisms, i.e., canonical transformations in a ``would be category'' of symplectic linear vector spaces, see~\cite{Weinstein:SympCat}. We rather postulate how gluing rules for two spacetime regions lead to reconstruction of solutions in a new region obtained by gluing old regions along a boundary hypersurface.

Gauge symmetries in the boundary may be encoded in this context as the degeneracy linear subspace $\ker \widetilde{\omega}_{\partial M}$ of the presymplectic structure. This follows from the horizontal exactness of the multisymplectic current contracted with gauge variations. See~\cite{Lee:LocSymm, Zuckerman:ActionPr}, although the roots of this argument may be traced back to Dirac. Hence, in this particular case gauge reduction takes the simple form of a linear space quotient. Throughout this work we consider that all our linear and affine spaces are the final result of gauge reduction. We regard $L_{\partial M}$, $A_{\partial M}$ as linear and affine symplectic spaces.

A suitable polarization requires a tame complex structure $J_{\partial M}$ on $L_{\partial M}$. This is the main ingredient determining the resulting quantum theory apart from the purely classical data described so far. $J_{\partial M}$ can be constructed using the Dirichlet--Neumann operator. For the abelian Yang--Mills case the latter was obtained in~\cite{Dia:dtonopab}.

In this section we proceed to the definition of the main objects of the theory and validate the axiomatic framework of Section~\ref{sec:classax}.

\subsection{Hypersurfaces}

We proceed in this subsection to prove semiclassical Axioms (C1), (C2), (C3). For each $(n-1)$-dimensional hypersurface $\Sigma$ we define the space $L_\Sigma$ as well as bilinear maps, $\omega_\Sigma(\cdot,\cdot),$ $[\cdot,\cdot]_\Sigma,$ $g_\Sigma(\cdot,\cdot),$ $\{\cdot,\cdot\}_\Sigma,$ and a complex structure $J_\Sigma$ mentioned in Axioms (C1), (C2), (C3). The complete validation concerning affine spaces $A_\Sigma$ and an affine $1$-form $\theta_\Sigma$ will be given subsequently.

Take the bilinear map
\begin{gather}\label{eqn:[.,.]}
[\phi_1,\phi_2]_{\Sigma}:= \int_{\Sigma}\phi^D_2\wedge\star_{\Sigma}\phi^N_1,
\end{gather}
where ${\star_\Sigma}$ is the Hodge star operator in $\Sigma$ and where each $\phi_1$, $\phi_2$ is a couple $\phi_i=\big(\phi^D_i,\phi^N_i\big)\in \big(\Omega^1(\Sigma)\big)^{\oplus2}$. $\Omega^1(\Sigma)$ denotes the $1$-forms in $\Sigma$. At this stage indices $D$, $N$ are just labels. When applied to hypersurfaces $\Sigma$ embedded into a region $M$ with embedding $i_\Sigma\colon \Sigma\subseteq M$ and to $1$-forms $\gf\in\Omega^1(M)$, then~$\gf^D$ means the Dirichlet boundary data~$i^*_\Sigma\gf$ of~$\gf$, while $\gf^N=\star_\Sigma i^*_\Sigma(\star \xd\gf)$ refers to the Neumann boundary data of~$\gf$. Thus $\gf\mapsto \phi=\big(\gf^D,\gf^N\big)$ defines a linear map $r_M(\gf)=\phi$, $r_M\colon \Omega^1(M)\rightarrow \big(\Omega^1(\Sigma)\big)^{\oplus 2}$.

The presymplectic linear form $\widetilde{\omega}_\Sigma$ is described by the antisymmetric part of the bilinear form~$[\cdot,\cdot]$
\begin{gather*}%\label{eqn:symplectic-structure}
 \widetilde{\omega}_{\Sigma}(\phi_1,\phi_2)= \tfrac{1}{2}\big([\phi_1,\phi_2]_{\Sigma}-[\phi_2,\phi_1]_{\Sigma}\big).
\end{gather*}

Define the complex structure $J_\Sigma\colon \Omega^1(\Sigma)^{\oplus 2}\rightarrow \Omega^1(\Sigma)^{\oplus 2}$ as
\begin{gather}\label{eqn:J}
 J_{\Sigma}\big(\phi^D,\phi^N\big)= \big(\phi^N,-\phi^D\big).
\end{gather}
We can also define a symmetric bilinear part $g_{\Sigma}(\cdot,\cdot)=2\tilde{\omega}_{\Sigma}(\cdot,J_\Sigma\cdot)$. More explicitly
\begin{gather*}%\label{eqn:g_partial M}
 g_{\Sigma}({\phi_1},{\phi_2})= \int_{\Sigma}\phi_1^D\wedge\star_{\Sigma}\phi_2^D+ \phi_1^N\wedge\star_{\Sigma}\phi_2^N.
\end{gather*}

\begin{lma} The degeneracy space of the bilinear antisymmetric form $\widetilde{\omega}_\Sigma$ is
\begin{gather*}
 \ker\widetilde{\omega}_\Sigma= \big\{ \xd f\oplus \xd g\colon f,g\in \big(\Omega^0(\Sigma)\big)\big\}.
\end{gather*}
\end{lma}

\begin{proof}For the proof consider the $g_{\Sigma}$-orthogonal Hodge decomposition
\begin{gather*}%\label{eqn:Hodge}
 \Omega^k(\Sigma) = \xd \Omega^{k-1}(\Sigma)\oplus \mathfrak{H}^k(\Sigma)\oplus \xd^{\star_\Sigma}\Omega^{k+1}(\Sigma),
\end{gather*}
where $ \mathfrak{H}^k(\Sigma) := \big\{\lambda\in\Omega^k(\Sigma)\colon \xd\lambda=0=\xd^{\star}\lambda\big\}$ stands for harmonic $k$-forms. Thus
\begin{gather*}
 \ker \xd^{\star_\Sigma}= \big\{ \phi\in \Omega^1(\Sigma) \colon \xd^{\star_{\Sigma}}\phi = 0 \big\}
 = \mathfrak{H}^k(\Sigma)\oplus \xd^{\star_\Sigma}\Omega^{k+1}(\Sigma) \simeq \Omega^k(\Sigma)/\xd\Omega^{k-1}(\Sigma) .
\end{gather*}
Hence $\widetilde{\omega}_\Sigma (\phi,\xd f_1\oplus \xd f_2 )=0$ for every $\phi=\phi^D\oplus \phi^N\in (\ker \xd^{\star_\Sigma})^{\oplus 2}$ and exact $1$-forms $\xd f_1,\xd f_2\in\Omega^1(\Sigma)$. On the other hand, the equality $\widetilde{\omega}_\Sigma (\phi,\xd f_1\oplus \xd f_2)=0$ for every $f_1\oplus f_2\in \Omega^0(\Sigma)^{\oplus 2}$ implies
\begin{gather*}
 \int_\Sigma \xd f_1\wedge\star_\Sigma \phi^N- \phi^D\wedge\star_\Sigma \xd f_2 = \int_\Sigma f_1 \xd^{\star_\Sigma}\phi^N -f_2\xd^{\star_\Sigma}\phi^D, \qquad \forall\, f_1,f_2\in\Omega^0(\Sigma),
\end{gather*}
hence, $\xd^{\star_\Sigma}\phi^N=0 =\xd^{\star_\Sigma}\phi^D$.
\end{proof}

The complex structure defines the scalar multiplication by $\im$, hence we may define a Hermitian product $\{\cdot, \cdot\}_\Sigma = g_\Sigma+2\im\omega_\Sigma $ which in the complex linear vector space $(\ker \xd^{\star_\Sigma})^{\oplus 2}$ is non-degenerate.

Since in our proposed axiomatic framework we need Hilbert spaces, we rather complete the spaces of boundary data, i.e., we consider the Hilbert space $L^2\Omega^1(\Sigma)$ obtained as the $g_{\Sigma}$-closure of $\Omega^1(\Sigma)$. In order to see that the closure space of the Dirichlet boundary data $\phi^D$ on $\Sigma$ is in fact in $ L^2\Omega^1(\Sigma)$ take the closure $\overline{\ker \xd^{\star_\Sigma}}$ in the Sobolev space $W^{1,2}\Omega^1(\Sigma)$. Then the quotient norm of the gauge quotient for Dirichlet conditions $\overline{\ker \xd^{\star_\Sigma}}/\overline{\ker \xd^{\star_\Sigma}}\cap \xd W^{2,2}\Omega^0(\Sigma)$ is isomorphic to the $g_\Sigma$-norm in $L^2\Omega^1(\Sigma)$. On the other hand, the Neumann datum $\phi^N$ is contained in $\overline{\ker \xd^{\star_\Sigma}}\subseteq L^2\Omega^1(\Sigma)$. Thus we consider the Hilbert space obtained by the $g_\Sigma$-completion of the first-order boundary data on $\Sigma$ as the space of pairs of $1$-forms
\begin{gather}\label{eqn:L_Sigma}
 L_\Sigma := \big(\overline{\ker \xd^{\star_\Sigma}} / \overline{\ker \xd^{\star_\Sigma}}\cap \xd W^{2,2}\Omega^0(\Sigma)\big)
 \oplus \overline{\ker \xd^{\star_\Sigma}} \subseteq L^2\Omega^1(\Sigma)\oplus L^2\Omega^1(\Sigma).
\end{gather}
In (\ref{eqn:L_Sigma}) we have used the fact that in $\Sigma$ the differential has the following domains and rank
\begin{gather*}
\xd\colon \ W^{1,2}\Omega^1(\Sigma_{\geps})\rightarrow L^{2}\Omega^1(\Sigma_{\geps}).
\end{gather*}
This follows from Gaffney's inequality, see \cite{Sc},
\begin{gather}\label{eqn:Gaffney}
 \|\omega\|_{W^{1,2}}^2 \leq C\big( \|\xd\omega\|_{L^2}^2 + \|\xd^\star\omega\|_{L^2}^2 + \|\omega\|^2_{L^2} \big),\qquad
 \forall\, \omega\in W^{s,p}\Omega^1(M).
\end{gather}

Whenever we consider an abstract hypersurface $\Sigma $ we consider it as an $(n-1)$-manifold together with an $n$-dimensional Riemannian metric defined in a cylinder
\begin{gather*}
 \Sigma_\geps\cong \Sigma\times [0,\geps].
\end{gather*}
We recall some definitions of Sobolev spaces. First recall Hodge--Morrey--Friedrichs (HMF) decompositions for any Riemannian manifold $\Sigma_\geps$, with boundary $\partial \Sigma_\geps=\Sigma\sqcup\Sigma'$ with inclusion of the bottom component $i_{\Sigma}\colon \Sigma\rightarrow \Sigma_\geps$, see~\cite{Sc},
\begin{gather}
 \Omega^k(\Sigma_\geps)=\xd\Omega_D^{k-1}(\Sigma_\geps)\oplus\mathfrak{H}^k_N(\Sigma_\geps)\oplus
 \big(\mathfrak{H}^k(\Sigma_\geps)\cap \xd\Omega^{k-1}{\Sigma_\geps}\big)\oplus \xd^{\star}\Omega_N^{k+1}(\Sigma_\geps),\nonumber\\
 \Omega^k(\Sigma_\geps)=\xd\Omega_D^{k-1}(\Sigma_\geps)\oplus\mathfrak{H}^k_D(\Sigma_\geps)\oplus
 \big(\mathfrak{H}^k(\Sigma_\geps)\cap \xd^\star\Omega^{k-1}{\Sigma_\geps}\big)\oplus \xd^{\star}\Omega_N^{k+1}(\Sigma_\geps),\label{eqn:HMF}
\end{gather}
where
\begin{gather*}
 \Omega_D^k(\Sigma_\geps) := \big\{\alpha\colon \alpha\in\Omega^{k}(\Sigma_\geps)\colon i^*_{\partial \Sigma_\geps}\alpha=0\big\} , \\
 \Omega^k_N(\Sigma_\geps) := \big\{\beta \colon \beta\in\Omega^{k}(\Sigma_\geps)\colon i^*_{\partial \Sigma_\geps} (\star\beta )=0\big\} , \\
 \mathfrak{H}_D^k(\Sigma_\geps) := \mathfrak{H}^k(\Sigma_\geps)\cap \Omega^k_D(\Sigma_\geps), \qquad
 \mathfrak{H}_N^k(\Sigma_\geps) := \mathfrak{H}^k(\Sigma_\geps)\cap \Omega^k_N(\Sigma_\geps).
\end{gather*}
$\star$ is the Hodge operator on $\Sigma_\geps$.

Recall also that the differential $\xd$ acts on the chain complex $ \Omega_D^k(\Sigma_\geps)$, meanwhile the co\-dif\-ferential $\xd^{\star}$ acts on the complex~$\Omega_N^k(\Sigma_\geps)$. The space of harmonic forms, $\mathfrak{H}^k(\Sigma_\geps)$, is infinite-dimensional. Nevertheless its boundary conditioned subspaces, $ \mathfrak{H}^k_N(\Sigma_\geps),\mathfrak{H}^k_D(\Sigma_\geps)\subset \mathfrak{H}^k(\Sigma_\geps)$ are finite-dimensional.
For Sobolev spaces in decomposition (\ref{eqn:HMF}) we substitute the orthogonal summands of $W^{s-1,p}\Omega^k(\Sigma_\geps)$ by
\begin{gather*}
\xd W^{s,p}\Omega_D^{k-1}(\Sigma_\geps) = \big\{ \alpha\in \operatorname{dom} \xd \subseteq W^{s,p}\Omega^{k-1}(\Sigma_\geps) \colon
 i^*_{\partial \Sigma_\geps}\alpha=0 \big\}, \\
 W^{s-1,p}\mathfrak{H}_N^k(\Sigma_\geps) = W^{s-1,p}\mathfrak{H}^k(\Sigma_\geps)\cap W^{s-1,p}\Omega^k_N(\Sigma_\geps),\\
 W^{s-1,p}\mathfrak{H}^k(\Sigma_\geps) = \big\{ \lambda \in \operatorname{dom}\xd\cap \operatorname{dom} \xd^\star \subseteq W^{s-1,p}\Omega^k(\Sigma_\geps)
 \colon \xd \lambda=0=\xd^{\star}\lambda \big\}, \\
 \xd^{\star}W^{s,p}\Omega^{k+1}_N(\Sigma_\geps) = \big\{ \beta\in \operatorname{dom}\xd^\star \subseteq W^{s,p}\Omega^{k+1}(\Sigma_\geps)\colon
 i^*_{\partial \Sigma_\geps} (\star\beta)=0 \big\} .
\end{gather*}
Notice that by finite-dimensionality $W^{s-1,p}{\mathfrak H}_N^k(\Sigma_\geps)={\mathfrak H}^k_N(\Sigma_\geps)$.

Consider a principal bundle ${\cal E}(\geps)=i^*_{\Sigma}{\cal E}$ where ${\cal E}$ is a principal bundle over $\Sigma_\geps$. We fix base points $\eta_{{\cal E}(\geps)}$ for the affine spaces $\operatorname{Conn}({\cal E})$ of connections $A_{\Sigma_\geps}$ modulo gauge, as follows.

 \begin{dfn}\label{Definition4.2}
We say that a fixed base point $\eta_{{\cal E}(\geps)}\in \operatorname{Conn}({\cal E})$ of the affine space of connections of~${\cal E}$, $\operatorname{Conn}({\cal E})$ is a {\em local minimum} if it satisfies the following conditions:
\begin{enumerate}\itemsep=0pt
\item %\label{item:EL}
$\eta_{{\cal E}(\geps)}$ is a solution of $\xd^\star F^{\eta_{{\cal E}(\geps)}}=0$. The {\em $($affine$)$ space of solutions with Lorentz gauge fixing} can be defined as
\begin{gather*}
 \big\{\eta\in \operatorname{Conn}({\cal E}(\geps))\colon \xd^\star F^\eta=0,\, \xd^\star \gf=0,\, \gf\in W^{3/2,2}\Omega^1 (\Sigma_\geps )
 \big\},
\end{gather*}
where we consider the decomposition $\eta=\eta_{{\cal E}(\geps)}+\xd\gf $. We consider this space as the space of representatives of gauge classes with the gauge defined by exact $1$-forms translations. The gauge quotient will be
\begin{gather*}
 A_{\Sigma_\geps}:=[\eta_{{\cal E}(\geps)}] + L_{\Sigma_\geps} ,
\end{gather*}
which is an affine space modeled over the corresponding linear space $L_{\Sigma_\geps}$ obtained as the topological quotient of
 \begin{gather*}
 \big\{\gf\in\Omega^1(\Sigma_\geps) \colon \xd^\star \gf=0,\, \gf\in W^{3/2,2}\Omega^1 (\Sigma_\geps )\big\},
 \end{gather*}
 by the exact translations $\xd W^{5/2,2}\Omega^0(\Sigma_\geps)$. For the choice of the Sobolev ratios $3/2$, $5/2$, see~(\ref{eqn:def C}) and the explanation below for the domains and range of Dirichlet and Neumann maps.

\item %\label{item:critical}
$\eta_{{\cal E}(\geps)}$ satisfies
\begin{gather*}
 \int_{\Sigma_\geps} F^{\eta_{{\cal E}(\geps)}}\wedge \star \xd\xi = 0
\end{gather*}
for every coexact $1$-form $\xi\in\Omega^1(\Sigma_\geps)$, $\xd^\star\xi=0$.
\end{enumerate}
\end{dfn}

This induces a base point
\begin{gather}\label{eqn:base-point1}
 \eta_{\Sigma,{\cal E}(\geps)} := \big( \big[\eta_{{\cal E}(\geps)}^D\big], \eta_{{\cal E}(\geps)}^N \big) \in A_\Sigma
\end{gather}
in the space of boundary conditions modulo gauge $A_\Sigma :=\big(\operatorname{Conn} ({\cal E}(\geps) ) /\xd W^{2,2}\Omega^{0}(\Sigma)\big)\times \overline{\ker \xd^{\star_\Sigma}}$. Define
\begin{gather}\label{eqn:theta_1}
 \overline{\theta}_\Sigma(\eta,\phi) := \int_\Sigma \phi^D\wedge i^*_\Sigma(\star F^{\eta}) .
\end{gather}
Then the following lemma holds.

\begin{lma}\label{lma:base-point1} Let $\eta_{{\cal E}(\geps)}$ be a local minimum. Then the following assertions hold.
\begin{enumerate}\itemsep=0pt
\item[$1.$]%\label{item:a1}
$\overline{\theta}_{\Sigma} ({\eta}_{{\cal E}({\geps})},\gf )= \int_\Sigma \gf^D\wedge i^*_\Sigma(\star F^{\eta_{{\cal E}(\geps)}})$ vanishes for every $\gf\in \ker \xd^\star$.

\item[$2.$]%\label{item:a2}
For the variation along the boundary $ \overline{\theta}_{\Sigma} (\eta,\gf )$ with $\eta={\eta}_{{\cal E}(\geps)}+\gf$, for $\gf\in\overline{\ker \xd^\star}$, we have
\begin{gather*}
 \overline{\theta}_{\Sigma}(\eta,\gf')= \int_{\Sigma}\gf^N \wedge\star_{\Sigma} (\gf')^D.
\end{gather*}

\item[$3.$]%\label{item:a3}
$\overline{\theta}_{\Sigma}$ is gauge invariant, where the gauge action in $A_{\Sigma_\geps}$ is by exact translations $\xd f\in \xd W^{5/2,2}\Omega^2(\Sigma)$, $\tilde{\eta}=\eta+ \xd f=(\eta_{{\cal E}(\geps)}+\gf)+\xd f$
\begin{gather*}
 \overline{\theta}_{\Sigma} (\tilde{\eta}, \gf' )= \int_{\Sigma}\gf^N\wedge\star_{\Sigma} (\gf')^D,
 \qquad \forall\, \gf'\in\ker \xd^\star.
\end{gather*}
Therefore $\overline{\theta}_{\Sigma}$ induces a symplectic potential $\theta_\Sigma$ for the translation invariant symplectic structure $\omega_\Sigma$ in $A_{\Sigma}$.

\item[$4.$]%\label{item:a4}
For any other local minimum base point $\eta'_{{\cal E}(\geps)}$, the assertions given above also hold.
\end{enumerate}
\end{lma}

\begin{proof}Let us consider condition~1. Recall that $\partial \Sigma_\geps=\Sigma\sqcup \Sigma'$ with $\overline{\Sigma'}\cong \Sigma$. Take $\gf\in \overline{\ker \xd^\star}$. By Lemma~\ref{lemma:Dirichlet-YM2} %in Appendix~\ref{sec:minimal-YM}
there exists a~$\gf_1\in\Omega^1(\Sigma_\geps)$ such that
\begin{gather*}
 \xd\star \xd\gf_1=0, \qquad \xd^\star\gf_1=0, \qquad i^*_\Sigma\gf_1=i^*_\Sigma\gf, \qquad i^*_{\Sigma'}\gf_1=0.
\end{gather*}
Then
\begin{gather*}
 \int_\Sigma i^*_\Sigma(\gf\wedge\star F^{\eta_{{\cal E}(\geps)}}) = \int_{\partial \Sigma_\geps}i^*_\Sigma \big( \gf_1
 \wedge \star F^{\eta_{{\cal E}(\geps)}} \big)\\
\hphantom{\int_\Sigma i^*_\Sigma(\gf\wedge\star F^{\eta_{{\cal E}(\geps)}})}{}=
 \int_{\Sigma_\geps}i^*_\Sigma\big( \xd \gf_1 \wedge \star F^{\eta_{{\cal E}(\geps)}} + \gf_1 \wedge \xd \star F^{\eta_{{\cal E}(\geps)}}
 \big) = \int_{\Sigma\times[0,\geps]}F^{\eta_{{\cal E}(\geps)}} \wedge \star \xd\gf_1 = 0 .
\end{gather*}
The last equality follows from the very definition of local minimum. Condition~2 follows from $
 F^{\eta}=F^{\eta_{\cal E}(\geps)}+ \xd \gf$ for $\eta=\eta_{{\cal E}(\geps)}+\gf$, $\xd^\star\gf=0$, and condition~1 so that
\begin{gather*}
\int_{\Sigma} i^*_{\Sigma} \big(\gf' \wedge \star F^{\eta}\big) = \int_{\Sigma} i^*_{\Sigma}\big( \gf' \wedge \star (F^{\eta_{{\cal E}(\geps)}}+ \xd \gf) \big) =\int_{\Sigma} i^*_{\Sigma} ( \gf' \wedge \star \xd \gf )\\
\hphantom{\int_{\Sigma} i^*_{\Sigma} \big(\gf' \wedge \star F^{\eta}\big)}{}=
\int_{\Sigma} i^*_{\Sigma} \gf' \wedge i^*_{\Sigma} (\star \xd \gf) =\int_{\Sigma} (\gf')^D \wedge \star_{\Sigma} \gf^N.
\end{gather*}
Condition~3 follows from the triviality of the gauge action in the Neumann component, $\tilde{\eta}^N= \eta^N$. To verify condition~4 take two local minima $\eta_{\cal E}$, $\eta_{\cal E}'$, with $\eta_{{\cal E}(\geps)}'=\eta_{{\cal E}(\geps)} +\gf$, $\xd^\star \gf=0.$ Then
\begin{gather*}
 \int_M \big( F^{\eta_{{\cal E}(\geps)}}-F^{\eta'_{{\cal E}(\geps)}} \big)\wedge \star \big( F^{\eta_{{\cal E}(\geps)}}-F^{\eta'_{{\cal E}(\geps)}} \big) = \int_M \big( F^{\eta_{{\cal E}(\geps)}}-F^{\eta'_{{\cal E}(\geps)}} \big) \wedge \star \xd\gf =0 ,
\end{gather*} by condition~2 of local minimum. Therefore $F^{\eta_{{\cal E}(\geps)}'}= F^{\eta_{{\cal E}(\geps)}}$. Hence
\begin{gather*}
 \overline{\theta}_{\Sigma}\big({\eta}'_{{\cal E}(\geps)}-{\eta}_{{\cal E}(\geps)},\xi\big) = 0
\end{gather*}
for every coclosed $\xi$ and
\begin{gather*}
 \overline{\theta}_{\Sigma}\big({\eta}_{{\cal E}(\geps)}', \xi\big) = \overline{\theta}_{\Sigma}\big({\eta}_{{\cal E}(\geps)}, \xi\big).\tag*{\qed}
\end{gather*}\renewcommand{\qed}{}
\end{proof}

Relation (\ref{eqn:[.,.]}) implies the following translation invariance condition
\begin{gather}\label{eqn:theta}
 \left[\phi_1,\phi_2\right]_{\Sigma} + \theta_{\Sigma}\left(\eta',\phi_2\right)= \theta_{\Sigma}\left(\phi_1+\eta',\phi_2\right),
 \qquad \forall\, \phi_1,\phi_2\in L_\Sigma,\quad \eta'\in A_\Sigma.
\end{gather}

Notice that here the notation $\eta'$ refers to gauge classes $[\eta]$ of boundary conditions $\eta=\big(\eta^D,\eta^N\big)$ in~$\Sigma$. This provides the definition of the affine spaces $A_\Sigma$ as well as the definition of~$\theta_\Sigma$ and completes the proof of Axioms~(C1),~(C3). A~change on orientation $\overline{\Sigma}$ in $\Sigma$ proves~(C2). Axiom~(C7) is just the formal definition of an infinitesimal cylinder with slice region.

\subsection{Regions and hypersurfaces}

In this subsection we prove the validity of semiclassical Axioms~(C4), (C5), (C6). Axiom~(C8) follows from the very definition of YM action~(\ref{eqn:YM-action}).

For a space-time region $M$ with boundary $\partial M$, the space of boundary conditions in the boundary {\em cylinder}, $L_{\partial M}$ is defined as in the hypersurface case as~(\ref{eqn:L_Sigma}). We now consider the subspace $L_{{M},\partial M}\subseteq L_{\partial M}$ of {\em topologically admissible boundary conditions}, as the image of the continuous linear map~$r_M$ of Dirichlet--Neumann boundary conditions. This verifies Axiom~(C5). In principle, $r_M$ should have domain in $W^{3/2,2}\Omega^1(M)\cap \overline{\ker \xd^\star}/ W$. Nevertheless, the normal vanishing components in the first HMF decomposition~(\ref{eqn:HMF}) yield the gauge fixing subspace
\begin{gather*}
 \overline{\ker \xd^\star} = \mathfrak{H}^1_N(M)\oplus \xd^{\star}W^{5/2,2}\Omega_N^{2}(M) \subseteq W^{3/2,2}\Omega_N^{1}(M) ,
\end{gather*}
that is orthogonal to exact forms and that has the same image $L_{M,\partial M}$. Hence we take
\begin{gather}\label{eqn:def C}
 r_M\colon \ \mathfrak{H}^1_N(M)\oplus \xd^{\star}W^{5/2,2}\Omega_N^{2}(M)  \rightarrow L_{\partial {M}}
\end{gather}
as the linear map of boundary conditions. Since we consider the (topological) gauge quotient of the $W^{3/2,2}$-closure of the gauge fixing space on the bulk~(\ref{eqn:Lorentz}) modulo translations by exact forms, $\xd W^{5/2,2}\Omega^0(M)$, then in the preimage we consider the quotient $W^{1/2,2}$-topology (coarser than the $W^{3/2,2}$-topology). Meanwhile on the codomain of~$r_M$ we consider the $L^2$-topology defined in~(\ref{eqn:L_Sigma}). Thus we verify Axiom~(C5).

More precisely, the Dirichlet and Neumann condition maps $\gf\mapsto \big(\gf^D,\gf^N\big)$ for any hypersurface $\Sigma\subseteq \partial M$, are defined by the inclusion $i_{\Sigma}\colon \Sigma\rightarrow M$ for linear spaces
\begin{align*}
 W^{3/2,2}\Omega^1(M) & \rightarrow W^{1,2}\Omega^1(\Sigma),\\
 \varphi &\mapsto \varphi^D=i_{\Sigma}^*\varphi ,\\
 W^{3/2,2}\Omega^1(M) & \rightarrow L^2\Omega^1(\Sigma),\\
\varphi &\mapsto \varphi^N=\star_{\partial M} i^*_{\Sigma}(\star \xd\varphi),
\end{align*}
restricted to $\mathfrak{H}^1_N(M)\oplus\xd^{\star}W^{5/2,2}\Omega_N^{2}(M)$. For the justification of the power $3/2$ in the Dirichlet map see \cite[Lemma~3.3.2]{Sc}. For the rank $L^2$-space in the Neumann condition map, we use a~condition generalizing Gaffney's inequality~(\ref{eqn:Gaffney}), see \cite[Lemma~2.4.10(iii)]{Sc}, namely
\begin{gather*}
 \|\omega\|_{W^{s,p}}\leq
 C\big(\|\xd\omega\|_{W^{s-1,p}} + \|\xd^\star \omega\|_{W^{s-1,p}} \big),\qquad \forall\, \omega\in W^{s,p}\Omega^1(M).
\end{gather*}
Hence when we take the gauge classes in $\Sigma$ we get $r_M\big(\big[\gf^D\big],\gf^N\big)\in L^2\Omega^1(\Sigma)^{\oplus 2}$ such that the following diagram of continuous linear maps commutes
\begin{gather*}%\label{eqn:Sobolev-diagram}
\xymatrix{
 \mathfrak{H}^1_N(M)\oplus \xd^{\star}W^{5/2,2}\Omega_N^{2}(M)
 \ar[r]^{\big(\gf^D,\gf^N\big)}
 \ar@{<->}[d]_{\cdot/\xd W^{2,2}\Omega^0(M)}
 \ar@{-->}[rd]^{r_M}
 &
 W^{1,2}\Omega^1(\Sigma)\oplus L^2\Omega^1(\Sigma)
 \ar@{->>}[d]^{\cdot/\xd W^{2,2}\Omega^0(\Sigma)}
 \\
 W^{1/2,2}\Omega^1(M)
 \ar[r]
 &
 L_{\Sigma}
 =
 L^2\Omega^1(\Sigma)^{\oplus 2} ,
}
\end{gather*}
where the downwards projections mean gauge quotients in the bulk and the boundary respectively.

On the other hand, the affine space $A_M$ consisting of solutions modulo gauge, is modeled over a linear space $L_M$ obtained as the {\em Lorentz gauge fixing on the bulk}
\begin{gather}\label{eqn:Lorentz}
 \big\{ \gf\in W^{3/2,2}\Omega^1(M) \colon \xd^\star \xd \gf =0 ,\, \xd^\star\gf=0 \big\},
\end{gather}
quotiented by exact $1$-forms acting by exact translations $\xd f\in \xd W^{5/2,2}\Omega^0(M)$.

For every first variation, $\xi \in W^{3/2,2}\Omega^1(M)$ (not necessarily a variation of solutions), for {Lorentz gauge fixing on the bulk}, $\xd^\star \xi=0$, take $\eta= \eta_{\cal E}+ t\xi$, $t\in(-\geps,\geps)$. Then
\begin{gather*}
 S_M(\eta)= S_M(\eta_{\cal E}) + t\int_MF^{\eta_{\cal E}}\wedge\star \xd\xi + \frac{t^2}{2} \int_M \xd\xi \wedge\star \xd\xi.
\end{gather*}
Condition 1 of Definition~\ref{Definition4.2}, integration by parts and Stokes' theorem together imply that the variation of the action in the boundary is
\begin{gather*}
 \xd S [\eta_{\cal} ](\xi) = \left.\frac{\xd}{\xd t}\right\vert_{t=0}S_M(\eta) =
 \int_MF^{\eta_{\cal E}}\wedge\star \xd\xi = \int_M \xd^\star F^{\eta_{\cal E}}\wedge\star\xi+
 \int_{\partial M} i^*_{\partial M} \xi \wedge i^*_{\partial M} \big(\star F^{\eta_{\cal E}}\big)\\
 \hphantom{\xd S [\eta_{\cal} ](\xi)}{} =
 \int_{\partial M} i^*_{\partial M} \xi \wedge i^*_{\partial M} \big(\star F^{\eta_{\cal E}}\big),
\end{gather*}
where $i_{\partial M}\colon \partial M\rightarrow M$ is the inclusion. Define
\begin{gather*}
 \theta_{\partial M} [\eta_{{\cal E}} ](\xi):= \int_{\partial M} i^*_{\partial M}\big( \xi \wedge \star F^{\eta_{\cal E}}\big).
\end{gather*}
See \cite{Woo:geomquant} to recall the definition of the boundary variation $\theta_{\partial M}$ for arbitrary field theories. If we consider Lemma~\ref{lma:base-point1} with $\Sigma=\partial M$ then condition~2 of Definition~\ref{Definition4.2} implies
\begin{gather*}
 \theta_{\partial M}\left[\eta_{{\cal E}}\right](\xi)=0
 ,\qquad \forall\, \xi \in \overline{\ker d^\star}\subseteq W^{3/2,2}\Omega^1(M) .
\end{gather*}
Notice that this coincides with (\ref{eqn:theta_1}) with $\Sigma=\partial M$ and that it actually depends on the ima\-ges~$a_M(\eta_{\cal E})$,~$r_M(\xi)$. Thus we have
\begin{gather*} \theta_{\partial M} (a_M(\eta_{\cal E}),r_M(\xi) )= \theta_{\partial M} [\eta_{{\cal E}} ](\xi).
\end{gather*}

Furthermore, $ F^{\eta}=F^{{\eta}_{{\cal E}}}+\xd\gf$ is the orthogonal decomposition, where $F^{{\eta}_{\cal E}}$ is orthogonal to every exact $2$-form $\xd\gf$, with $\xd^\star \gf=0$ for every $[\eta]=[\eta_{\cal E}]+[\gf]\in A_M$. That is why $[\eta_{\cal E}]$ is a~local minimum. The local minimum representative $\eta_{\cal E}$ does not have to be unique in general. If $[\eta'_{\cal E}]\in A_M $ solves $F^{\eta_{\cal E}'}=F^{\eta_{\cal E}}$ then $\eta'_{\cal E} $ is also a local minimum representative.

Recall that when we consider boundary conditions consisting of pairs of $1$-forms $\big(\phi^D,\phi^N\big)$, we take the gauge quotient of the \emph{axial gauge fixing space on the boundary}~$\partial M$. This topological quotient yields the linear space $L_{\partial M}$ obtained as the closure of the pairs of coclosed $1$-forms, $\big(\big[\phi^D\big],\phi^N\big)\in L^2\Omega^1(\partial M)\oplus L^2\Omega^1(\partial M)$ just as in~(\ref{eqn:L_Sigma}). If we consider the closure for the $W^{1/2,2}$-topology (coarser than the $W^{3/2,2}$-topology) of the Lorentz gauge-fixing for solutions intersected with suitable summands in the second HMF decomposition~(\ref{eqn:HMF})
\begin{gather}\label{eqn:Lorentz-2}
 \{\xd^\star \xd\gf=0\}\cap \big(\mathfrak{H}^1_N(M)\oplus \xd^{\star}\Omega_N^{2}(M)\big),
\end{gather}
then (\ref{eqn:Lorentz-2}) is isomorphic to the space $L_M$ of solutions modulo gauge. Therefore to obtain the boundary conditions $L_{\tilde{M}}=r_M(L_M)\subseteq L_{M,\partial M}$, we can take just the $r_M$-image of (\ref{eqn:Lorentz-2}). For the linear map $r_M\colon L_M\rightarrow L_{\partial M}$, $ r_M(\varphi) = \big(\big[\varphi^D\big], \varphi^N\big)$, where $\gf^D$ and $\gf^N$ are the Dirichlet and Neumann boundary conditions respectively, we consider the corresponding affine map $a_M\colon A_M\rightarrow A_{\partial M}$, with image $A_{\tilde{M}}$.

We summarize the results we have obtained so far. The space ${L}_{\partial M}$ depends just on $\partial M$ and the Riemannian metric of the cylinder $\partial M\times [0,\geps]$ and does \emph{not} depend on the topology of~$M$. The subspace ${L}_{M,\partial M}$ depends on the germ of the metric on the boundary, {\em and} on the relative topology of $M$ and $\partial M$. For YM fields the inclusion $A_{M,\partial M}\subseteq A_{\partial M}$ is proper in general. More explicitly, ${L}_{\tilde{M}}\subseteq L_{M,\partial M}$ is the Lagrangian subspace obtained by the closure of
\begin{gather*}%\label{eqn:Lagrangian}
 \big\{ \big(\big[\gf^D\big], \gf^N\big)\in {L}_{M,\partial M}\colon \gf\in \big(\mathfrak{H}_N^1(M)\oplus \xd^{\star}W^{5/2,2}\Omega^2_N(M)\big)\cap \{\xd^\star \xd\gf\} \big\}.
\end{gather*}
Recall that the harmonic components of $\gf^D$ for all solutions $\gf$ define a finite codimension space of ${\mathfrak H}^1(\partial M)\simeq H^1_{\rm dR}(\partial M)$.

Let ${\cal E}_{\partial M}=i^*_{\partial M} \cal E$ be the induced principal bundle on $\partial M\subseteq M$. Take the fixed base point
\begin{gather}\label{eqn:base-point0}
 \eta_{M,{\cal E}}:=\big(\eta^D_{\cal E},\eta^N_{\cal E}\big)
 \in \operatorname{Conn} ({\cal E}_{\partial M} )\times \overline{\ker \xd^{\star_{\partial M}}}.
\end{gather}
Its gauge class $a_{M} ({\eta}_{{\cal E}} )=\big(\big[\eta^D_{\cal E}\big],\eta^N_{\cal E}\big)$ is also well defined in
\begin{gather*}
 A_{\partial M}:= \big(\operatorname{Conn} ({\cal E}_{\partial M} ) /\xd W^{2,2}\Omega^0(\partial M)\big) \times
 \overline{\ker \xd^{\star_{\partial M}}}.
\end{gather*}
The Dirichlet condition, $\eta^D_{\cal E}\in\operatorname{Conn} ({\cal E}_{\partial M} )$ is given by the connection in the induced bundle on $\operatorname{Conn} ({\cal E}_{\partial M} )$ defined by ${\eta}_{{\cal E}}\in \operatorname{Conn}({\cal E})$. Here $\star_{\partial M}$ denotes the Hodge star for the induced metric on~$\partial M$. On the other hand, the Neumann condition ${\eta}_{{\cal E}}^N \in \overline{\ker \xd^{\star_{\partial M}}}\subseteq L^2\Omega^1(\partial M)$ is defined as the derivative $\dot{{\eta}^\tau}(0)$ of a one parameter family
\begin{gather*}
 \overline{\eta}^\tau = \eta_{\cal E}^D+ \overline{\gf}^\tau\in \operatorname{Conn} ({\cal E}_{\partial M} ),
 \qquad \overline{\eta}^0=\eta^D_{\cal E}.
\end{gather*}
That is,
\begin{gather*}
 \eta^N_{\cal E}:=\left.\frac{\xd}{\xd\tau}\right\vert_{\tau=0}\overline{\gf}^\tau\in \overline{\ker \xd^{\star_{\partial M}}}\subseteq L^2\Omega^1(\partial M).
\end{gather*}
See Lemma~\ref{lemma:Dirichlet-YM} for further details.

For $\Sigma=\partial M$ and $\phi_i=r_M(\gf_i)$, $i=1,2$, we have (\ref{eqn:[.,.]}) and the translation invariance condition~(\ref{eqn:theta}). Notice that we have two base points for~$A_{\partial M}$. On the one hand, a~base point is defined for $\partial M$ as boundary component of a cylinder $\partial M_\geps$, $\eta_{\partial M, {\cal E}(\geps)}$, as in~(\ref{eqn:base-point1}). On the other hand one is defined on the boundary of~$M$, $\eta_{M,{\cal E}}$, as in~(\ref{eqn:base-point0}). Then
\begin{gather*}
 {\eta_{ {\cal E}(\geps)}} - {\eta_{{\cal E}}}=\gf, \qquad F^{\eta_{ {\cal E}(\geps)}} - F^{\eta_{{\cal E}}}=\xd\gf,
 \qquad \xd^\star\gf=0,\\
 \int_{(\partial M)_\geps} \big(F^{\eta_{ {\cal E}(\geps)}} - F^{\eta_{{\cal E}}}\big) \wedge \star \big(F^{\eta_{ {\cal E}(\geps)}} - F^{\eta_{{\cal E}}}\big)
 = \int_{(\partial M)_\geps} \big(F^{\eta_{ {\cal E}(\geps)}} - F^{\eta_{{\cal E}}}\big) \wedge \star \xd\gf = 0.
\end{gather*}
Thus $F^{\eta_{{\cal E}}}\vert_{(\partial M)_\geps}=F^{\eta_{ {\cal E}(\geps)}}$. Hence both $\eta_{{\cal E}(\geps)}$ as well as $\eta_{\cal E}$ are local minima for the YM action for the cylinder $(\partial M)_\geps$. The translation invariance condition~(\ref{eqn:theta}) holds for $\Sigma=\partial M$.

We verify Axiom (C6). Notice that
\begin{gather*} \int_M\xd^\star\gf\wedge\star \xd^\star\gf = \int_M \xd\xd^\star\gf\wedge\star\gf +
 \int_{\partial M}\xd^\star\wedge\star(\nu\lrcorner\gf\vert_{\partial M})\\
 \hphantom{\int_M\xd^\star\gf\wedge\star \xd^\star\gf}{}= \int_M\gf\wedge\star \xd\xd^\star\gf - \int_{\partial M}\gf^D\wedge\star(\nu\lrcorner \xd\gf\vert_{\partial M}) =
 - \int_{\partial M}\gf^D\wedge\star_{\partial M}\gf^N,
\end{gather*}
where $\nu$ is the normal vector to~$\partial M$, and $\nu\lrcorner \xd\gf\vert_{\partial M}=(-1)\star_{\partial M}i^*_{\partial M}(\star \xd\gf)$.

Let ${\eta}_{{\cal E}}$ be a fixed point so that for every connection $\eta$, its curvature is $F^\eta=F^{{\eta}_{{\cal E}}}+\xd\gf$ with $\gf\in\Omega^1(M)$ and $\xd\gf$ orthogonal to $F^{{\eta}_{{\cal E}}}$.

Let $\eta-{\eta}_{{\cal E}}=\gf$, $\eta'-{\eta}_{{\cal E}}=\gf'$, then for the YM action
\begin{gather*}
 S_M(\eta)-S_M(\eta') = \frac{1}{2}\int_M \xd\gf\wedge\star \xd\gf -\frac{1}{2}\int_M \xd\gf'\wedge\star \xd\gf' \\
\hphantom{S_M(\eta)-S_M(\eta')}{}= - \frac{1}{2}\int_{\partial M}\gf^D\wedge\star_{\partial M}\gf^N + (\gf')^D\wedge\star_{\partial M}(\gf')^N.
\end{gather*}
On the other hand
\begin{gather*}
 \theta_{\partial M}(a_M(\eta),r_M(\eta-\eta'))+ \theta_{\partial M}(a_M(\eta'),r_M(\eta-\eta'))\\
\qquad{}=
 \int_{\partial M}\gf^D\wedge\star_{\partial M}(\gf-\gf')^N + \int_{\partial M}(\gf')^D\wedge\star_{\partial M}(\gf-\gf')^N.
\end{gather*}
Since $\gf$, $\gf'$ are tangent to a Lagrangian subspace,
\begin{gather*}
 -\int_{\partial M}\gf^D\wedge\star_{\partial M}(\gf')^N +\int_{\partial M}(\gf')^D\wedge\star_{\partial M}\gf^N=0.
\end{gather*}
Therefore,
\begin{gather*}%\label{eqn:S_M-theta}
 S_M(\eta)-S_M(\eta') = -\tfrac{1}{2}\theta_{\partial M}(a_M(\eta),r_M(\eta-\eta')) - \tfrac{1}{2}\theta_{\partial M}(a_M(\eta'),r_M(\eta-\eta'))
\end{gather*}
holds $\forall\, \eta,\eta'\in {A}_M$. This proves (\ref{eq:actsympot}) and completes the verification of Axiom~(C6).

\subsection{Gluing}\label{sec:3}

For the process described in the semiclassical Axiom (C9), suppose that a region $M_1$ is obtained from a region $M$ by gluing along $\Sigma,\Sigma'\subset \partial M$. We also suppose that the principal bundle ${\cal E}=p^*_{MM_1}{\cal E}_1$ on $M$ is obtained from the principal bundle ${\cal E}_1$ on $M_1$. From the projection map $p_{MM1}\colon M\rightarrow M_1$ we get the inclusion $p^*_{MM_1}\colon \Omega^1(M_1)\rightarrow \Omega^1(M)$. We obtain maps
\begin{gather*}
 r_{M;\Sigma\overline{\Sigma'}}=p^*_{MM_1}\vert_{{L}_{M_1}}\colon \ {L}_{M_1} \rightarrow {L}_{M} .
\end{gather*}
For the affine spaces $A_M=[\eta_{\cal E}] + L_M$ and $A_{M_1} = [\eta_{{\cal E}_1}]+ L_{M_1}$, with local minima base points $\eta_{\cal E}$, $\eta_{{\cal E}_1}$ respectively. The corresponding affine linear map
$ a_{M;\Sigma\overline{\Sigma'}}\colon {A}_{M_1} \rightarrow {A}_{M}$, yields $a_{M;\Sigma\overline{\Sigma'}}([\eta_{{\cal E}_1}])=[\eta_{\cal E}]$.

From the decomposition
\begin{gather*}
 L_{\partial M} = L_{\Sigma_1} \oplus L_\Sigma \oplus L_{\overline{\Sigma'}} ,
\end{gather*}
we get projections from $L_{\partial M}$ onto $L_\Sigma$, $L_{\overline{\Sigma'}}$. From the gluing isometry $f_\Sigma\colon \Sigma\rightarrow \overline{\Sigma'}$, there is a~commuting diagram
\begin{gather*}\xymatrix{
 L_{M_1}
 \ar@{^{(}->}[r]^{ r_{M;\Sigma\overline{\Sigma'}}}
 &
 L_{ M}
 \ar[r]^{r_{M;\Sigma}} \ar[rd]_{r_{M;\overline{\Sigma'}}}
 &
 L_{\Sigma}
 \\
 &&
 L_{\overline{\Sigma'}}.
 \ar[u]_{f_\Sigma^*}
} \end{gather*}
Here
\begin{gather*}
 f_\Sigma^* \circ {r_{M;\overline{\Sigma'}}} ([\gf_1]) = \big(f^*_\Sigma \gf^D_1,-f^*_\Sigma \gf^N_1\big) = {r_{M;\Sigma}}([\gf_1]) ,
\end{gather*}
for $[\gf_1]\in L_{M_1}$. Then by the uniqueness of boundary conditions $r_{M_1;\Sigma_1}([\gf_1])$, associated to solutions $[\gf_1] \in L_{M_1}$ we have that
\begin{gather*}
 r_{M;\Sigma} \circ r_{M;\Sigma\overline{\Sigma'}}([\gf_1])=r_{M;\Sigma}([\gf_1]),\qquad
 r_{M;\overline{\Sigma'}} \circ r_{M;\Sigma\overline{\Sigma'}}(\gf_1)=r_{M;\overline{\Sigma'}}([\gf_1]).
\end{gather*}

Consider the map, $\gf\mapsto i_{\partial M_1}^*\gf$, $i_{\partial M_1}\colon \partial M_1\rightarrow \partial M$, arising from the inclusion
 $\partial M_1\subset \partial M$. It induces a linear map $\lambda_1 \colon {L}_{M,\partial M}\rightarrow {L}_{M_1,\partial M_1}$,
 \begin{gather*}
 \lambda_1\big( \big[\gf^D\big]\oplus \gf^N \big) = \big(i_{\partial M_1}^*\big[\gf^D\big]\oplus i_{\partial M_1}^*\gf^N\big) ,
 \end{gather*}
 and a commuting diagram of linear maps
\begin{gather*}\xymatrix{
 {L}_{M_1}
 \ar[rrr]^{r_{M;\Sigma\overline{\Sigma'}}} \ar[dr] \ar@{-->}[dd]_{r_{M_1}}
 &&&
 {L}_{M}
 \ar[dl] \ar@{-->}[dd]^{r_M}
 \\
 &
 {L}_{M_1,\partial M_1}
 \ar@^{{(}->}[dl]
 &
 {L}_{M,\partial M}
 \ar@^{{(}->}[dr] \ar@{->>}[l]^{\lambda_1}
 &
 \\
 {L}_{\partial M_1}
 &&&
 {L}_{\partial M},
 \ar@{^{(}->}[lll]^{i^*_{\partial M_1}}
}\end{gather*}
while for the corresponding affine spaces ${A}_{M,\partial M}=\big(\big[\eta_{\cal E}^D\big],\eta_{\cal E}^N\big)+{L}_{M,\partial M}$, ${A}_{\partial M}=\big(\big[\eta_{\cal E}^D\big],\eta_{\cal E}^N\big)+{L}_{\partial M}$ we have
\begin{gather*}\xymatrix{
 {A}_{M_1}
 \ar[rrr]^{a_{M;\Sigma\overline{\Sigma'}}} \ar[dr] \ar@{-->}[dd]_{r_{M_1}}
 &&&
 {A}_{M}
 \ar[dl] \ar@{-->}[dd]^{r_M}
 \\
 &
 {A}_{M_1,\partial M_1}
 \ar@^{{(}->}[dl]
 &
 {A}_{M,\partial M}
 \ar@^{{(}->}[dr] \ar@{->>}[l]^{\alpha_1}
 &
 \\
 {A}_{\partial M_1}
 &&&
 {A}_{\partial M}.
 \ar@{^{(}->}[lll]^{i^*_{\partial M_1}}
}\end{gather*}
Recall that $[\eta_{\cal E}]$ is a local minimum base point in $A_M$ while $\eta^D_{\cal E},$ and $\eta^N_{\cal E}$ are its Dirichlet and Neumann boundary data respectively. This yields a fixed point $\big(\big[\eta^D_{\cal E}\big],\eta^N_{\cal E}\big)$ for $A_{M,\partial M}$.

Recall that $J_{\partial M}$ is densely defined in $L_{\partial M}=L^2\Omega^1(\partial M)$. By taking the $L^2$-closure we have a continuous linear map $J_{\partial M}\colon L_{\partial M}\rightarrow L_{\partial M}$. As observed in~\cite{Dia:dtonopab} there is a densely $J_{\partial M}$ linear subspace consisting of smooth topological admissible boundary conditions. By taking the $L^2$-closure we induce a continuous linear map $J_{\partial M}\colon {L}_{M,\partial M}\rightarrow{L}_{M,\partial M}$. Thus, ${L}_{M,\partial M}$ is a complex subspace of $({L}_{\partial M},J_{\partial M})$, hence symplectic. Notice that the general Dirichlet--Neumann map (D-N) for classical Hodge Laplacian BVP on functions is a map with domain $W^{s,2}(\partial M)$ and rank $W^{s-1,2}(\partial M)$. Nevertheless, we used the D-N operator for $1$-forms \emph{modulo gauge actions}. This yields an operator continuous with respect to $L^2$-topologies both in the domain and the rank. This gauge freedom is an important difference between the treatment of BVP for $1$-forms with respect to the BVP for functions.

\section{Special cases}\label{sec:examples}

In this section we specialize to the case where regions are $2$-manifolds. Then, hypersurfaces are disjoint unions of circles. What is more, the data $L_{\Sigma}$ for a circle $\Sigma$ consists of just a~pair of real numbers $\phi=\big(\phi^D_\Sigma,\phi^N_\Sigma\big)$. Recall that Yang--Mills solutions on a region~$M$ satisfy that their curvature~$F^\eta$ is $f\cdot \xd\mu_{M}$ where $f$ is a constant and $\xd\mu_{M}$ the area form on~$M$, see~\cite{Atiyah-Bott-1983}.

\subsection{The disc}
Take a Riemannian $2$-dimensional region $B$ which is the injective image of a ball of radius $\delta>0$ under the exponential map with center $P\in B$. We call this region the \emph{disc} $B$ of center $P$ and geodesic radius $\delta>0$.

We now describe the space $A_B$ of YM solutions in a disc $B$ of radius $\delta>0$ and with boundary component $\partial B=\Sigma$. For any YM solution, $\eta\in A_B$ may be explicitly described in polar coordinates $(r, \theta)$, $r\in [0,\delta]$, $\theta\in [0,2\pi)$, as $\eta(\theta,r)=\frac{\eta^D_{\Sigma}}{\delta^2} r^2\xd\theta$. Along $\Sigma$ the Dirichlet condition is $\eta(\theta,\delta)=\eta^D_{\Sigma}\xd\theta$. For the Neumann condition notice that $\frac{\partial }{\partial r}\eta(\geps,\theta)=\frac{2\eta^D_{\Sigma}}{\delta}\xd\theta$ implies $\eta^N_{\Sigma}=2\eta^D_{\Sigma}/\delta$ and $\xd\eta=2\eta^D/\delta^2 r \xd r \xd\theta$ is proportional to the usual area form $\xd\mu_{B}=r\xd r \xd\theta$.

For the general Riemannian case we consider a disc $B$ of geodesic radius $\delta>0$ and center of the exponential map $P$, see for instance~\cite{Gallot-2004}. Take the coordinates $(x,\tau)$ of $B-\{P\}$ obtained from the collar along $\Sigma=\partial B$, $X_{\Sigma}\colon \Sigma\times[0,\delta)\rightarrow B$. With these coordinates the center is $P=\lim\limits_{\tau\rightarrow\delta}X_{\Sigma}(x,\tau)$ while $\Sigma^\tau=X_{\Sigma}(\Sigma,\tau)$ is the image under the exponential map of the circle of radius $\delta-\tau>0$. The area form in the disc can be then written as $\xd\mu_{B} = {\rm length}\left(\Sigma^{\tau}\right) \xd\tau\wedge \xd x$. Here $x\in[0,1]$ is a coordinate in $\Sigma^{\tau}$ proportional to the arc-length coordinate, with proportionality constant
\begin{gather*}
{\rm length}\big(\Sigma^{\tau}\big)= 2\pi (\delta-\tau)\left(1-\frac{K(P)}{6}(\delta-\tau)^2+o\big((\delta-\tau)^2\big)\right),
\end{gather*}
where $K(P)$ means curvature. Hence, we have the explicit solution
\begin{gather*}
 \eta(x,\tau)\vert_{B} = \eta^D_{\Sigma} \left( \frac{ \int^{\tau}_\delta {\rm length}(\Sigma^\tau)\xd\tau }{\int^{0}_\delta {\rm length}(\Sigma^\tau)\xd\tau } \right) \xd x,\\
 \frac{\xd}{\xd \tau}\eta(x,\tau)= \eta^D_{\Sigma}\frac{{\rm length}(\Sigma^\tau)}{\int^{0}_\delta {\rm length} (\Sigma^\tau)\xd\tau}\xd x,
\end{gather*}
such that $\xd \eta = \frac{\eta^D}{\int_\delta^0 {\rm length}(\Sigma^\tau) \xd\tau} \xd \mu_B$, and the boundary conditions
\begin{gather*}
 \eta(x,0)\vert_{B}=\eta^D_{\Sigma}\xd x, \qquad \frac{\xd}{\xd\tau}\eta(x,0)\vert_{B}=\eta^N_{\Sigma}\xd x
\end{gather*}
satisfy $\eta^N_{\Sigma}=\eta^D_{\Sigma}\frac{{\rm length}(\Sigma)}{\int^{0}_\delta {\rm length} (\Sigma^\tau)\xd\tau}$. Therefore, the condition required for every YM solution holds, namely $ F^\eta = f\xd\mu_{B}$ for a constant~$f$ with
\begin{gather*}
 f= \eta^D_{\Sigma} / \int_\delta^0 {\rm length}(\Sigma^\tau) \xd\tau =
 \eta^N_{\Sigma} / {\rm length}(\Sigma).
\end{gather*}
This completes the description of the space $A_B$ of all YM solutions in the disc~$B$. Notice that for the disc~$B$, since any principal bundle ${\cal E}$ is trivialized by the choice of a connection (by parallel transport along geodesics stemming out of the center of the disc), we have a natural identification of linear and affine spaces
\begin{gather*}
 A_B=L_B,\qquad A_{\partial B }=L_{\partial B}.
\end{gather*}
Hence, for $L_{\tilde B}\subseteq L_\Sigma=L_{B,\partial B}$ we consider the subspace of pairs $\big(\phi_\Sigma^D,\phi^N_\Sigma\big)\in L_\Sigma$ such that
\begin{gather*}
 \phi^N_\Sigma\cdot {\rm length}(\Sigma)= \phi^D_\Sigma\cdot \int_\delta^0{\rm length}(\Sigma^\tau) \xd\tau.
\end{gather*}

\subsection[Surface of genus ${\rm g}\geq 2$ with boundary]{Surface of genus $\boldsymbol{{\rm g}\geq 2}$ with boundary}
Let $M$ be a connected $2$-manifold of genus ${\rm g}\geq 2$ and with $m\geq 1$ boundary components $\Sigma_1,\dots,\Sigma_m$. In this case
\begin{gather*}
 L_{\partial M}=L_{\Sigma_1}\oplus\cdots\oplus L_{\Sigma_m}\simeq \mathbb{R}^{2m}
\end{gather*}
 is a $2m$-dimensional linear space consisting of constant boundary data, $\phi=\big(\phi^D_{1},\phi^N_{1}\big)\oplus \cdots \oplus\big(\phi^D_{m},\phi^N_{m}\big)$. It has symplectic form
\begin{gather*}
 \omega_{\partial M}(\phi,\xi) =\sum_{i=1}^m\big(\phi^D_{i} \cdot \xi^N_i - \phi_{i}^N\cdot \xi^D_i\big)\cdot{\rm length}(\Sigma_j),
 \qquad \forall\, \phi,\xi\in L_{\partial M},
\end{gather*}
where ${\rm length}(\Sigma_i)$ is the total length of each boundary component $\Sigma_i$.

Take a principal bundle ${\cal E}$ over $M$. Recall that the base point $ \eta_{{\cal E}}\in A_{M}$ is chosen in such a~way that $F^{\eta_{{\cal E}}}$ is orthogonal to every exact one form $\xd\gf $ for every curvature form $F^\eta=F^{\eta_{\cal E}}+\xd\gf$ of a connection $\eta=\eta_{{\cal E}}+\gf \in A_{M}$, $\gf\in\Omega^1(M)$. Since the space of closed $2$-forms is one-dimensional, the curvature $F^\eta$ is already contained in the one-dimensional space generated by~$F^{\eta_{{\cal E}}}$, the curvature of the base point~$\eta_{{\cal E}}$. Therefore~$\gf$ is closed. This imposes a necessary condition for vectors in the $(2m-2)$-dimensional linear subspace~$L_{M,\partial M}$, namely
\begin{gather}\label{eqn:L_dM}
 \sum_{i=1}^m\phi^D_{\Sigma_i}\cdot {\rm length}(\Sigma_i)=0 ,\qquad \forall\, \big(\phi^D_{\partial M},\phi^N_{\partial M}\big)\in W.
\end{gather}
In fact, $L_{M,\partial M}$ consists of the maximal symplectic subspace of $L_{\partial M}$ contained in the codimension one (non-symplectic and coisotropic) subspace $W\subseteq L_{\partial M}$ defined by~(\ref{eqn:L_dM}). More precisely, $L_{M,\partial M}$ is isomorphic to $W /W^{\bot}$ where $W^{\bot}$ is the one-dimensional symplectic complement of $W$, with $W^{\bot}\subseteq W$. This ends the description of linear spaces $L_{M,\partial M}$, $L_{\partial M}$.

For the space $L_{M}$ we consider the space of flat connections which is parametrized by a linear space of dimension $\dim L_{M}=(2{\rm g}+m-1)$, see~\cite{Sengupta2002}. In fact, $L_{M}$ consists of two types of boundary data: on the one hand the data of the integral of the connection along the boundary components, $\phi_{\Sigma_1}^D,\dots,\phi^D_{\Sigma_m},$ satisfying condition~(\ref{eqn:L_dM}). On the other hand, the data of the connection integral along the $2{\rm g}$ non-contractible cycles in the interior of~$M$ generating~$H_1(M)$. Hence we have the Lagrangian
$L_{\tilde{M}} $ consisting of $ \big(\phi_{\Sigma_1}^D,0\big)\oplus\cdots \oplus\big(\phi^D_{\Sigma_m},0\big) \in L_{M,\partial M}$, satisfying~(\ref{eqn:L_dM}).

\subsection[Surface of genus ${\rm g}\geq 2$ without boundary]{Surface of genus $\boldsymbol{{\rm g}\geq 2}$ without boundary}
Let $M_1$ be a closed $2$-dimensional surface of genus ${\rm g}\geq 2$, obtained by gluing $m\geq 1$ discs of certain suitable small radii $\delta_i>0$, $i=1,\dots,m,$ along the boundary components of a surface~$M'$ of genus ${\rm g}\geq 2$ and $m$ boundary components. Thus we glue regions
\begin{gather*}
 M_1=M'\cup_\Sigma M'',\qquad M=M'\sqcup M'', \qquad M''=B_1\sqcup\cdots\sqcup B_m ,
\end{gather*}
along the boundaries
\begin{gather*}
 \Sigma=\partial M'=\sqcup_{i=1}^m \Sigma_i, \qquad \Sigma'=\partial M''= \sqcup_{i=1}^m \Sigma_i',\qquad \Sigma'_i=\partial B_i.
\end{gather*}
Let ${\cal E}_1$ be a fixed principal bundle on $M_1$ inducing the bundle ${\cal E}=p^*_{MM_1}{\cal E}_1$ in $M$. This means that we consider connections
\begin{gather*}
 \eta={\eta}_{\cal E}+\gf= a_{M;\Sigma,\overline{\Sigma'}}(\eta_1) \in A_M\subseteq \operatorname{Conn}({\cal E})
\end{gather*}
that are YM solutions induced by solutions $\eta_1=\eta_{{\cal E}_1}+\gf\in A_{M_1}$, $\gf\in\Omega^1(M_1)$. Since region~$M_1$ is a surface without boundary, hence $L_{{\partial M}_1}=L_{M_1,{\partial M}_1}=L_{\tilde{M}_1}=0$. From condition $\xd\gf =0$ (which implies the YM condition $\xd^\star \xd\gf=0$) it follows that $\dim L_{M_1}= 2{\rm g}=\dim H^1_{\rm dR}(M_1)$. Notice that the gauge condition is already incorporated in the cohomology class, $c({\cal E}_1)=\frac{1}{2\pi}[F^{\eta}]\in H^2_{\rm dR}(M_1)$, which is fixed.

Recall that for the discs we naturally identify
\begin{gather*}
 A_{M''}=L_{M''},\qquad A_{\partial M''}=L_{\partial M''}.
\end{gather*}
We also have $ L_{\partial M''}=L_{M'',\partial M''}$. Furthermore, there is an affine isomorphism $A_{\partial M'}\simeq A_{\partial M''}$ given by
\begin{gather*}%\label{eqn:iso}
 \eta_{\Sigma_i'}^D=-\eta_{\Sigma_i}^D,\qquad \eta_{\Sigma_i'}^N=-\eta_{\Sigma_i}^N.
\end{gather*}
Due to the natural identification $A_{\partial M'}=L_{\partial M'}$, there is an associate isomorphism $L_{\partial M''}\simeq L_{\partial M'}$.

Stokes' theorem explicitly imposes $m$ necessary conditions on $a_{M''}(\eta)\in A_{\tilde{M}''}$
\begin{gather}\label{eqn:Stokes}
 \eta^D_{\Sigma_i'}\cdot {\rm length}(\Sigma_i) = \eta^N_{\Sigma_i'}\cdot \int^{0}_{\delta_i} {\rm length}(\Sigma_i^\tau)d\tau ,\qquad i=1,\dots, m,
\end{gather}
for every solution $\eta\in A_{M''}$. Hence $L_{\tilde{M}''}\subseteq L_{\partial {M}''}$ is a subspace of dimension $m$ consisting of $\big(\phi^D_{\Sigma_1},\phi^N_{\Sigma_1}\big) \oplus\cdots\oplus\big(\phi^D_{\Sigma_m},\phi^N_{\Sigma_m}\big) \in L_{\partial M''}$, satisfying~(\ref{eqn:Stokes}).

From this follows the Lagrangian condition. That is, for each $\phi\in L_{\tilde{M}'}$
\begin{gather*}
 \sum_{i=1}^m\phi^D_{\Sigma_i'}\cdot \xi^N_i\cdot{\rm length}(\Sigma_i) = \sum_{j=1}^m\phi^N_{\Sigma_j'}\cdot \xi^D_j\cdot{\rm length}(\Sigma_j),
\end{gather*}
where $ \xi\in L_{\partial M''}$ satisfying the Lagrangian condition~(\ref{eqn:Stokes}).

Recall that for a surface $M'$ of genus ${\rm g}\geq 2$ and $m\geq 1$ boundary components $\dim L_{\tilde{M'}}=m-1$ while $\dim L_{{M'}}=2{\rm g} + m-1$. Each cohomology class in $H^1_{\rm dR}(M')$ is represented by a unique harmonic form in ${\mathfrak H}_N^1(M')$, while each class in $H^1_{\rm dR}(M',\partial M')$ is represented by a unique harmonic form in ${\mathfrak H}_D^1(M')$, yielding $2{\rm g}$ and $(m-1)$-dimensional subspaces of $L_{M'}$, respectively.

Let $\eta_1=\eta_{{\cal E}_1}+ \gf\in A_{M_1}$. The linear inclusion $r_{M;\Sigma,\overline{\Sigma'}}\colon L_{M_1}\rightarrow L_{M}$ maps $\gf$ onto
\begin{gather*}
 r_{M;\Sigma,\overline{\Sigma'}}(\gf)= i_{M'}^*\gf\oplus i_{M''}^*\gf \in L_{M'}\oplus L_{M''},
\end{gather*}
where $i_{M'}^*\gf\in L_{M'}$, $M'$ with
\begin{gather*}
 r_{M'}(i_{M'}^*\gf) = \big( \gf^D_{\Sigma_1},0 \big) \oplus\cdots\oplus \big( \gf^D_{\Sigma_m},0 \big)\in L_{\tilde{M}'},
 \qquad
 \sum_{i=1}^m\gf^D_{\Sigma_m}=0.
\end{gather*}
Meanwhile, $i_{M''}^*\gf \in L_{M''}=A_{M''}$ consists of YM solutions on discs, which may not be flat. If
\begin{gather*}
 a_{M''}(i_{M''}^*\eta_1) = \big(\phi^D_{\Sigma_1'},\phi^N_{\Sigma_1'}\big) \oplus\cdots\oplus
 \big(\phi^D_{\Sigma_m'},\phi^N_{\Sigma_m'}\big) \in A_{\tilde{M}''},
\end{gather*}
then the relation (\ref{eqn:Stokes}) holds for every pair of Dirichlet--Neumann data.

Recall that $F^{\eta_{{\cal E}_1}}= f \xd\mu_{M_1},$ with $ f= \phi^N_{\Sigma_i'}/ {\rm length}(\Sigma_i),
$ and $\xd\mu_{M_1}$ the area form in $M_1$, and that $ f=2\pi c({\cal E}_1) /{\rm area}(M_1)$ where $c({\cal E}_1)\in H^2_{\rm dR} (M_1;\mathbb{Z})\simeq \mathbb{Z}$. Hence, the boundary data~$\phi^D_{\Sigma_i'}$,~$\phi^N_{\Sigma_i'}$ can be obtained from the geometry of the boundary components and from the bundle~${\cal E}_1$.

Therefore the base point $\eta_{M',\cal E}= a_{M'}(i_{M'}^*\eta_{{\cal E}_1})=a_{M'}(\eta_{\cal E}) \in A_{\tilde{M}'}\subseteq L_{\partial M'}$ equals
\begin{gather*}
 \big({-}\phi^D_{\Sigma_1'},-\phi^N_{\Sigma'_1}\big) \oplus\cdots\oplus \big({-}\phi^D_{\Sigma'_m},-\phi^N_{\Sigma'_m}\big)
\end{gather*}
and the affine space $A_{\tilde{M}'}$ is the translation of $L_{\tilde{M}'}${\samepage
\begin{gather*}
 A_{\tilde{M}'}=\eta_{ M',\cal E}+ L_{\tilde{M}'}\subseteq L_{\partial M'}\simeq L_{\partial M''}
\end{gather*}
(which has null Neumann components) prescribed by ${\cal E}_1$ with translation vector $a_{M'}(i_{M'}^*\eta_{{\cal E}_1})$.}

In particular, if we take $m-1 = 2{\rm g}$, then we have the exact amount of degrees of freedom in the boundary conditions space $L_{\tilde{M}'}\subseteq L_{M',\partial M'}$ to parametrize solutions $\gf_1\in L_{M_1}$. Thus
\begin{gather*}
 m-1=\dim L_{M_1} =\dim L_{\tilde{M}'}=\dim L_{\tilde{M}''}-1.
\end{gather*}

This completes the description of the closed surface without boundary.

The case of the sphere $M_1$ obtained by gluing two discs $M'$, $M''$ should be treated differently, but the construction can also be performed.

We now calculate the amplitude corresponding to a closed surface $M_1$ of genus ${\rm g}\geq 2$. We decompose $M_1$ as described in our previous discussion. The cases of genus ${\rm g}=0,1$ will also be excluded, but by the quantum setting, see Section~\ref{subsec:torus} below.

To calculate the anomaly factor $c(M;\Sigma,\overline{\Sigma'})$, considering the linear theory will suffice. For every $\gf_{1}\in L_{M_1}$ the induced solutions in $M',$ and $M''$ satisfy $\gf_{\Sigma}=-\gf_{\Sigma'}$ and
\begin{gather*}
 K_\xi(\gf_{\Sigma})\overline{K_{\xi}(\gf_{\Sigma'})}= \exp\big( \tfrac{1}{2} \big( \{\xi,\gf_{\Sigma} \}_{\partial M'} - \overline{ \{\xi,\gf_{\Sigma} \}_{\partial M'} } \big)\big)=1,
\end{gather*}
where we used coisotropy for solutions. Thus, the anomaly factor can be calculated as in~(\ref{eq:lanom}) by
\begin{gather}\label{eq:explicit-c}
 \int_{\hat{L}_{\Sigma}} \rho_M^{\rm L} \big(K_\xi\otimes\overline{K_\xi} \big) \xd\nu_{\Sigma} =1.
\end{gather}
Therefore, since $L_{\tilde{M}_1}=0$, the vacuum amplitude $ \rho_{M_1} \big(\hat{K}_{\eta_1} \big)$ would be independent of the choice of the fixed base point connection~$\eta_1$ in the cohomology class $c({\cal E}_1)$.

Related to this example, according to~\cite{Oe:freefermi}, if $M_1$ is obtained by gluing of $M'$ and $M''$ along the boundary $\Sigma\cong\overline{\Sigma'}$ where $\partial M'=\Sigma_0 \sqcup \Sigma$ and $\partial M''\cong \Sigma' \sqcup \Sigma_2$, then $ c(M;\Sigma,\overline{\Sigma'}) =1$.

If $f$ denotes the constant factor of the curvature $F^{\eta_{{\cal E}_1}}=f\cdot d\mu_{M_1}$ with respect to the area form, then the Hodge dual $f^\star\in \mathbb{R}$ satisfies $f^\star f=4\pi^2 (c({{\cal E}_1}) )^2 $. Hence $ \rho_{M_1}\big(\hat{K}_{\eta_1}\big) = \exp({\rm i}S_{M_1}(\eta_{\mathcal{ E}_1}))$ equals
\begin{gather}\label{eqn:vacuum-for-surfaces}
 \exp\big( \tfrac{\rm i}{2}f\cdot f^*\cdot {\rm area}(M_1) \big) = \exp\big({\rm i} 2\pi^2 {\rm area}(M_1)\cdot (c({{\cal E}_1}))^2 \big).
\end{gather}

The partition function (as the amplitude is also called in the case of empty boundary) for quantum YM theory on compact closed surfaces has been extensively studied in several quantization frameworks where the Hilbert spaces ${\cal H}_\Sigma$ associated to circles $\Sigma$ are generated by class functions of the structure group. In the axiomatic GBQFT this deduction was reported in~\cite{Oe:2dqym}. In the present work geometric quantization yielded another prescription for Hilbert spaces. The precise relationship of these two quantization frameworks needs to be clarified. In particular we need to clarify how~(\ref{eqn:vacuum-for-surfaces}) is related to the expression
\begin{gather*}
 \sum_{n\in\mathbb{Z}} \exp\big( {-}\tfrac{1}{2}{\rm area}(M_1)\cdot n^2 \big)
\end{gather*}
appearing for instance in~\cite{BlTh:qym,Wit:qgauge2d}.

In the torus case the partition function yields infinite trace. We may ask also what is the precise relationship between the divergence of the partition function and the anomaly factor $c(M;\Sigma,\overline{\Sigma'})$.

\subsection{The torus}\label{subsec:torus}
We contrast the case of gluing two isometric regions, with the case of a cylinder, where $c(M{;}\Sigma,\overline{\Sigma'})$ yields infinity and the integrability condition required in Theorem~\ref{thm:lgluing} is not satisfied. We consider $M=\Sigma\times [0,1]$ with the gluing of the bottom and top hypersurfaces $\Sigma\cong \Sigma\times\{0\}$, $\Sigma'=\Sigma\times\{1\}$. This includes the case of a closed surface $M_1$ of genus ${\rm g}=1$.

Any Dirichlet condition together with a null Neumann condition on $\Sigma$ determine completely solutions $\eta_1\in L_{M_1}$ in the whole $M_1$ after gluing. More explicitly, we have $\eta(x,\tau)=\eta^D_\Sigma(x)$, $\tau\in[0,1] $ for every coclosed $\eta^D_\Sigma\in\Omega^1(M')$.

Hence $\eta^D_{1,\Sigma}=\eta^D_{1,\Sigma'},$ $\eta^N_{1,\Sigma}=\eta^N_{1,\Sigma'}=0.$ Therefore $\eta_{1,\Sigma}=\eta_{1,\Sigma'}$. Thus using a projection $\hat{L}_{\tilde{M}}\rightarrow\hat{L}_{M_1}$ induced by the inclusion $r_{M;\Sigma,\overline{\Sigma}'}\colon L_{M_1}\rightarrow L_{M}$ we have
\begin{gather*}
\rho_M^{\rm L}\big(K_\xi\otimes\overline{K_{\xi}}\big) = \int_{\hat{L}_{\tilde{M}_1}} \exp\big(\tfrac{1}{2}g_{\Sigma}\big(\xi^{\rm R},\eta_{1,\Sigma}\big)\xd\nu_{\tilde{M}_1}\big).
\end{gather*}
By replacing in (\ref{eq:explicit-c}), $\rho_{M}^{\rm L} (K_\xi \otimes \iota_\Sigma(K_{\xi}))$ equals $\exp \big(\frac{1}{2} g_\Sigma \big(\xi^{\rm R},\xi^{\rm R}\big)\big)$, which is {\em not} $\xd\nu_\Sigma$-integrable.

\subsection{Complex manifolds} For the abelian Yang--Mills theory on some complex surfaces, namely~$\mathbb{CP}^2$, $\mathbb{CP}^1\times\mathbb{CP}^1$, there are also formulae for partition functions in the Euclidean action convention after Wick rotation, see~\cite{Verlinde-1995}. The precise process of how we can obtain these quantum calculations from its classical counterparts will be treated elsewhere. We envisage the use of a gluing process similar to that given for surfaces.

\section{Outlook}\label{sec:outlook}

We have presented in this work a functorial quantization prescription for abelian Yang--Mills theory on Riemannian manifolds targeting general boundary quantum field theory (GBQFT). There are some obvious directions for generalization which we comment on in the following.

{\bf Corners.} We have restricted ourselves here to only consider hypersurfaces that are closed, i.e., that do not have boundaries. However, this considerably restricts the possibilities for gluing regions. In particular, from a physics point of view it is important to be able to glue two regions with the topology of a ball to a new region that also has the topology of a ball (compare the discussion in \cite{Oe:gbqft,Oe:2dqym}). In order to accomplish this the gluing has to proceed along pieces of the boundary that are not connected components. That is, the gluing hypersurface has itself a boundary, usually referred to as a \emph{corner}. Corners are already allowed in the semiclassical axioms of Section~\ref{sec:classax} and in the functorial quantization scheme for affine field theory of Section~\ref{sec:quantization}. What is more, for abelian Yang--Mills theory implementing the symplectic reduction (Section~\ref{sec:classdem}) on the spaces~$A_\Sigma$ and~$L_{\Sigma}$ of germs of solutions on a hypersurface~$\Sigma$ appears to generalize straightforwardly to the case that~$\Sigma$ has a boundary. It remains to see that the complex structure generalizes nicely, but formula~(\ref{eqn:J}) is quite suggestive.

{\bf Lorentzian manifolds.} To describe physically realistic theories (such as electromagnetism) we need to work on Lorentzian manifolds. This introduces several complications. On the one hand the boundary value problem is hyperbolic instead of elliptic. Also, three different signatures for induced metrics can occur on hypersurfaces depending on them being spacelike, timelike or null. Finally, the complex structure has to be defined in a different way. For spacelike hypersurfaces there is considerable guidance from the literature, see, e.g.,~\cite{AsMa:qfieldscurved}. For timelike hypersurfaces there are some very basic examples~\cite{Oe:holomorphic,Oe:freefermi}. The general situation is open.

{\bf Non-abelian gauge theory.} In two dimensions non-abelian Yang--Mills theory is solvable and its quantization has been described in a TQFT-type formalism~\cite{BlTh:qym,Fin:qymriem,Rus:gauge2d, Wit:qgauge2d}. The generali\-za\-tion to include corners was described in~\cite{Oe:2dqym}. It is clear from the latter work that and how in the case with corners the axioms of the quantum theory (Appendix~\ref{sec:coreaxioms}) have to be modified. Whether this modification is also sufficient in higher dimensions is unclear. On the classical side the symplectic reduction on hypersurfaces becomes more complicated than in the abelian theory, at least when corners are present. However, working out the modifications compared to the abelian theory appears feasible. What in much more problematic in dimensions higher than two is the description of the dynamics within regions. The solution spaces become complicated manifolds and it is unknown what they are in general. Thus, it is quite unclear how to perform the step of the quantization that consists in constructing the amplitude maps. Of course this problem is not special to the present framework of GBQFT but appears in all approaches to quantum field theory when trying to deal with non-linear theories.

\appendix

\section{Geometric data}\label{sec:geomax}

We recall here the formalization of the notion of spacetime in terms of a \emph{spacetime system} on which both the classical axioms (Section~\ref{sec:classax}) and the quantum axioms (Appendix~\ref{sec:coreaxioms}) are based. The presented version is adapted from \cite[Section~4.1]{Oe:freefermi}.

There is a fixed positive integer $d$, the \emph{dimension} of spacetime. We are given a collec\-tion~$\sts_0^{\textrm{c}}$ of connected oriented topological manifolds of dimension~$d$, possibly with boundary, that we shall refer to as \emph{$($connected$)$/$($regular$)$ regions}. Furthermore, there is a collection~$\sts_1^{\mathrm{c}}$ of connected oriented topological manifolds of dimension $d-1$, possibly with boundary, that we shall refer to as \emph{$($connected$)$ hypersurfaces}. The manifolds are either abstract manifolds or they are all concrete closed regular submanifolds of a~given fixed \emph{spacetime manifold}. In the former case we call the spacetime system \emph{local}, in the latter we call it \emph{global}.

There is a notion of formal \emph{disjoint union} both for regular regions and for hypersurfaces. This leads to the collection $\sts_0$, of all formal finite unions of elements of $\sts_0^{\mathrm{c}}$, and to the collection~$\sts_1$, of all formal finite unions of elements of $\sts_1^{\mathrm{c}}$. In the local case, the unions may be realized concretely as actual disjoint unions. In the global case, only unions with members whose interiors are disjoint are allowed in $\sts_1$ and $\sts_0$. Note that in this case the elements of $\sts_1$ and $\sts_0$ cannot in general be identified with submanifolds of the spacetime manifold as overlaps on boundaries may occur. We call all members of $\sts_0$ \emph{$($regular$)$ regions} and all members of~$\sts_1$ \emph{hypersurfaces} and extend the notion of disjoint union to them.

The collections $\sts_1$ and $\sts_0$ are closed under orientation reversal. Also, any boundary of a~regular region is a~hypersurface. That is, taking the boundary defines a map $\partial\colon \sts_0\to\sts_1$. If we want to emphasize explicitly that a given manifold is in one of those collections we use the attribute \emph{admissible}.

\looseness=1 There is a notion of \emph{gluing} of elements, both of $\sts_0$ and of $\sts_1$. To avoid confusion we prefer for hypersurfaces the term \emph{decomposition}, reserving \emph{gluing} for regions. Given a presentation of a~hypersurface $\Sigma$ as the union of hypersurfaces $\Sigma_1,\dots,\Sigma_n$ we call this a decomposition if (a)~each $\Sigma_i$ is closed in $\Sigma$ and (b)~the intersection of $\Sigma_i$ with $\Sigma_j$ is contained in their boundaries for each~$i$ and~$j$ with $i\neq j$. These intersections are called \emph{corners}. Throughout the present article, we do not allow for corners. That is, we require these intersections to be empty.

It is convenient to also introduce the concept of a \emph{slice region}. A slice region is topologically a hypersurface, but thought of as an infinitesimally thin region. Concretely, the slice region associated to a hypersurface $\Sigma$ will be denoted by $\hat{\Sigma}$ and its boundary is defined to decompose as the disjoint union $\partial \hat{\Sigma}=\Sigma\cup\overline{\Sigma}$. There is one slice region for each hypersurface. We refer to regular regions and slice regions collectively as \emph{regions}.

The notion of \emph{gluing} of regions is as follows. Suppose we are given a region $M$ with its boundary decomposing as the union $\partial M=\Sigma_1\cup\Sigma\cup\overline{\Sigma'}$, where $\Sigma'$ is a copy of $\Sigma$. ($\Sigma_1$ may be empty.) Then, we may obtain a new region $M_1$ by gluing $M$ to itself along~$\Sigma$,~$\overline{\Sigma'}$. That is, we identify the points of $\Sigma$ with corresponding points of $\Sigma'$ to obtain~$M_1$. The resulting region $M_1$ might be inadmissible, in which case the gluing is not allowed.

Depending on the model to be considered, the manifolds may carry additional structure. It is common in particular that the hypersurfaces need to be ``thickened'', i.e., are equipped with germs of ambient $d$-manifolds. Also, the additional structure of a metric or other types of vector bundles or principal bundles are common. In that case all the hypersurfaces and regions are equipped with this additional structure and the different operations such as orientation reversal or gluing need to be compatible with the additional structures. This might also entail that additional data must be specified when a gluing is performed. In the present work the additional structures present are those of a Riemannian metric and of a principal bundle. In Section~\ref{sec:examples} examples are shown where associated additional data required for gluing need to be made explicit.

\section{Quantum axioms}\label{sec:coreaxioms}

A GBQFT on a spacetime system is a model satisfying the following axioms \cite{Oe:holomorphic}.
\begin{itemize}\itemsep=0pt
\item[(T1)] Associated to each hypersurface $\Sigma$ is a complex separable Hilbert space $\cH_\Sigma$, called the \emph{state space} of
 $\Sigma$. We denote its inner product by $\langle\cdot,\cdot\rangle_\Sigma$. If $\Sigma$ is the empty set, $\cH_{\Sigma}=\C$.
\item[(T1b)] Associated to each hypersurface $\Sigma$ is a conjugate linear isometry $\iota_\Sigma\colon \cH_\Sigma\to\cH_{\overline{\Sigma}}$. This map is an involution in the sense that $\iota_{\overline{\Sigma}}\circ\iota_\Sigma$ is the identity on $\cH_\Sigma$.
\item[(T2)] Suppose the hypersurface $\Sigma$ decomposes into a union of hypersurfaces $\Sigma=\Sigma_1\cup\cdots\cup\Sigma_n$. Then, there is an isometric isomorphism of Hilbert spaces \smash{$\tau_{\Sigma_1,\dots,\Sigma_n;\Sigma}\colon \! \cH_{\Sigma_1}\ctens{\cdots}\ctens\cH_{\Sigma_n}\!\to\!\cH_\Sigma$}.
\item[(T2b)] The involution $\iota$ is compatible with the above decomposition. That is, $\tau_{\overline{\Sigma}_1,\dots,\overline{\Sigma}_n;\overline{\Sigma}} \circ(\iota_{\Sigma_1}\ctens\cdots\ctens\iota_{\Sigma_n}) =\iota_\Sigma\circ\tau_{\Sigma_1,\dots,\Sigma_n;\Sigma}$.
\item[(T4)] Associated to each region $M$ is a linear map from a dense subspace $\cH_{\partial M}^\ds$ of the state space~$\cH_{\partial M}$ of its boundary $\partial M$ to the complex numbers, $\rho_M\colon \cH_{\partial M}^\ds\to\C$. This is called the \emph{amplitude} map.
\item[(T3x)] Let $\Sigma$ be a hypersurface. The boundary $\partial\hat{\Sigma}$ of the associated slice region $\hat{\Sigma}$ decomposes into the union $\partial\hat{\Sigma}=\overline{\Sigma}\cup\Sigma'$, where $\Sigma'$ denotes a second copy of $\Sigma$. Then, $\tau_{\overline{\Sigma},\Sigma';\partial\hat{\Sigma}}(\cH_{\overline{\Sigma}}\tens\cH_{\Sigma'})\subseteq\cH_{\partial\hat{\Sigma}}^\ds$. Moreover, $\rho_{\hat{\Sigma}}\circ\tau_{\overline{\Sigma},\Sigma';\partial\hat{\Sigma}}$ restricts to a bilinear pairing $(\cdot,\cdot)_\Sigma\colon \cH_{\overline{\Sigma}}\times\cH_{\Sigma'}\to\C$ such that $\langle\cdot,\cdot\rangle_\Sigma=(\iota_\Sigma(\cdot),\cdot)_\Sigma$.
\item[(T5a)] Let $M_1$ and $M_2$ be regions and $M= M_1\sqcup M_2$ be their disjoint union. Then $\partial M=\partial M_1\sqcup \partial M_2$ is also a disjoint union and $\tau_{\partial M_1,\partial M_2;\partial M}\big(\cH_{\partial M_1}^\ds\tens \cH_{\partial M_2}^\ds\big)\subseteq \cH_{\partial M}^\ds$. Moreover, for all $\psi_1\in\cH_{\partial M_1}^\ds$ and $\psi_2\in\cH_{\partial M_2}^\ds$,
\begin{gather*}
 \rho_{M}\circ\tau_{\partial M_1,\partial M_2;\partial M}(\psi_1\tens\psi_2)= \rho_{M_1}(\psi_1)\rho_{M_2}(\psi_2) .
% \label{eq:glueaxa}
\end{gather*}
\item[(T5b)] Let $M$ be a region with its boundary decomposing as a union $\partial M=\Sigma_1\cup\Sigma\cup \overline{\Sigma'}$, where~$\Sigma'$ is a copy of $\Sigma$. Let $M_1$ denote the gluing of~$M$ with itself along $\Sigma$, $\overline{\Sigma'}$ and suppose that~$M_1$ is a region. Note $\partial M_1=\Sigma_1$. Then, $\tau_{\Sigma_1,\Sigma,\overline{\Sigma'};\partial M}(\psi\tens\xi\tens\iota_\Sigma(\xi))\in\cH_{\partial M}^\ds$ for all $\psi\in\cH_{\partial M_1}^\ds$ and $\xi\in\cH_\Sigma$. Moreover, for any ON-basis $\{\xi_i\}_{i\in I}$ of $\cH_\Sigma$, we have for all $\psi\in\cH_{\partial M_1}^\ds$,
\begin{gather}
 \rho_{M_1}(\psi)\cdot c(M;\Sigma,\overline{\Sigma'}) =\sum_{i\in I}\rho_M\circ\tau_{\Sigma_1,\Sigma,\overline{\Sigma'};\partial M} (\psi\tens\xi_i\tens\iota_\Sigma(\xi_i) ),
\label{eq:glueid}
\end{gather}
where $c(M;\Sigma,\overline{\Sigma'})\in\C\setminus\{0\}$ is called the \emph{gluing anomaly factor} and depends only on the geometric data.
\end{itemize}

\section{Geometry of abelian Yang--Mills fields}\label{sec:minimal-YM}

We present minimal tools of differential geometry to deal with abelian gauge classical Yang--Mills (YM) fields.

Suppose that $M$ has Riemannian metric $h_{ab}$, $a,b=1,\dots,n$. Consider a hypersurface \smash{$\Sigma\subseteq M$} with induced a Riemannian metric $\overline{h}_{ij}=h_{ij}$, $i,j=1,\dots,n-1$. For a geodesic tubular neighborhood $\Sigma_\geps\subseteq M$ we consider the immersion $X_\Sigma \colon \Sigma\times [0,\geps]\rightarrow\Sigma_\geps\subseteq M$ where geodesics $\tau\mapsto X_\Sigma(x,\tau)$ remain orthogonal to the hypersurfaces $\Sigma^\tau:=X_\Sigma(\Sigma,\tau)\subseteq M$ that describe the evolution of~$\Sigma$. Here $x\in \Sigma$ and $\tau\in[0,\geps]$. Using this foliation with leafs $\Sigma^\tau$, decompose $1$-forms on $\Sigma_\geps$ as $\gf=\gf^\tau+\gf_\tau \xd\tau$.

First let us consider the expression
\begin{gather*}
 \star \xd\gf= \sum_{a,b=1}^{n-1} h^{ai}h^{bj} \left(\frac{\partial \gf_i}{\partial x^j}-\frac{\partial \gf_j}{\partial x^i}\right)\mu_{ab}
 + \sum_{i=1}^{n-1}h^{ai} \left(\frac{\partial \gf_i}{\partial \tau} -\frac{\partial \gf_\tau}{\partial x^i} \right)\mu_{a\tau},
\end{gather*}
where we used that $h^{\tau i}=\delta^\tau_i$ and $\xd\mu=\mu_{ab}\wedge \xd x^a \wedge \xd x^b$ with $\xd\mu=\sqrt{\det (h_{ab})}\xd x^1\wedge \dots\wedge \xd x^{n-1}\wedge \xd\tau$ the volume form of~$h_{ab}$ (notice that we adopt the notation $x^n=\tau$). For $\alpha=\star \xd\gf$, condition $\xd\star \xd\gf=0$ reads as $\xd\alpha=0$.

For the restriction onto $\Sigma^\tau$, we have
\begin{gather*}
 \overline{\gf}^\tau:= \big(X^\tau_\Sigma\big)^*\gf=\sum_{i=1}^{n-1} \gf_i(x,\tau)\xd x^i\in\Omega^1(\Sigma),
 \qquad \overline{\gf}_\tau:= \big(X^\tau_\Sigma\big)^*\gf_\tau
\end{gather*}
with $X^\tau_\Sigma=X_\Sigma(\cdot,\tau)$ and
\begin{gather*}
 \big(X^\tau_\Sigma\big)^*(\star \xd\eta)=
 \sum_{a=1}^{n-1}\overline{h^{ai}}(\tau)
 \left(\frac{\partial \gf_\tau}{\partial x^i}-\frac{\partial \gf_i}{\partial \tau}\right)\mu_a(\tau).
\end{gather*}
We denote the Riemannian metric induced on $\Sigma^\tau$ as $\overline{h_{ij}}^\tau$, $i,j=1,\dots,n-1$, and $\mu_{a}\wedge \xd x^a=\xd\mu$, $\mu_a(\tau):=(X^\tau_\Sigma)^*\mu_a$. Since
\begin{gather*}
 \big(X^\tau_\Sigma\big)^*(\star \xd\gf)= \star_{\Sigma^\tau}\left(\xd\overline{\gf}_\tau - \frac{\xd}{\xd\tau}\overline{\gf}^\tau\right),
\end{gather*}
then the YM condition, $\xd\star \xd\gf=0$, implies that
\begin{gather}\label{eqn:YM-decomposition}
 \xd^{\star_{\Sigma^\tau}}\left(\frac{\xd}{\xd\tau}\overline{\gf}^\tau\right) = \xd^{\star_{\Sigma^\tau}} \xd\overline{\gf}_\tau .
\end{gather}
On the other hand, we can decompose the divergence
\begin{gather*}
 \xd^\star \gf = -\sum_{k,l=1}^{n}\left( h^{kl}\frac{\partial \gf_k}{\partial x^l} - \sum_{j=1}^{n}\Gamma^j_{kl}\gf_j\right)
\end{gather*}
as
\begin{gather*}
 -\sum_{k,l=1}^{n-1}\left( h^{kl}\frac{\partial \gf_k}{\partial x^l} - \sum_{j=1}^{n-1}\Gamma^j_{kl}\gf_j\right)
 -\sum_{l=1}^{n-1}\left( h^{\tau l}\frac{\partial \gf_\tau}{\partial x^l} - \sum_{j=1}^{n-1}\Gamma^j_{\tau l}\gf_\tau\right)
 +\Gamma^\tau_{\tau l}\gf_\tau + \frac{\partial \gf_\tau}{\partial \tau}.
\end{gather*}
Hence
\begin{gather}\label{eqn:div-decomposition}
 (X_\Sigma^\tau)^*\xd^\star\gf = \frac{\partial \overline{\gf}_\tau}{\partial \tau} + \xd^{\star_{\Sigma^\tau}}\overline{\gf}^\tau,
\end{gather}
where we use that $h^{\tau l}=\delta^\tau_l$ and that for geodesics the Christoffel symbols $\Gamma^\cdot_{\tau\cdot}$ vanish for Fermi coordinates adapted to that geodesic.

For the linear gauge action $\tilde{\gf}=\gf+ \xd f$ we have
\begin{gather*}
 \tilde{\gf}^\tau= \gf^\tau + \sum_{i=1}^{n-1} \frac{\partial f}{\partial x^i} \xd x^i,
 \qquad \tilde{\gf}_\tau= \left(\gf_\tau+\frac{\partial f}{\partial \tau}\right)\xd\tau.
\end{gather*}
Choosing an {\em axial gauge fixing} in $\Sigma_\geps$ means solving $\tilde{\gf}_\tau=0$ or equivalently $\partial_\tau f = c_\tau -\gf_\tau$ where the constant~$c_\tau$ may be chosen in such a way that $\tilde{\gf}_\tau\vert_{\tau=0}=0$. Thus the YM condition in~(\ref{eqn:YM-decomposition}) can be rewritten as
\begin{gather*}%\label{eqn:YM-axial}
 \xd^{\star_{\Sigma^\tau}} \left(\frac{\xd}{\xd\tau}\overline{\tilde{\gf}}^\tau \right)=0 .
\end{gather*}
In particular, the Neumann boundary condition
\begin{gather*}
 \dot{{\gf}}(0):=\left(\frac{\xd}{\xd\tau}\overline{\tilde{\gf}}^\tau\right)\bigg\vert_{\tau=0}
\end{gather*}
is divergence-free. Furthermore, the Dirichlet boundary condition may also be gauge adjusted as divergence-free. Just solve the Poisson--Laplace equation, $\xd^{\star_\Sigma}\xd f^0=-(X^0_\Sigma)^*\xd^*\gf$, for $f^0=f\vert_\Sigma$. Therefore the divergence decomposition~(\ref{eqn:div-decomposition}) yields
\begin{gather*}%\label{eqn:div-axial}
 \xd^{\star_{\Sigma}} \overline{\tilde{\gf}}^0=0.
\end{gather*}

Choosing a {\em Lorentz gauge fixing} in $\Sigma_\geps$ means obtaining $\xd^\star \tilde{\gf}=0$ by solving $\frac{\partial \overline{\tilde{\gf}}_\tau}{\partial \tau} + \xd^{\star_{\Sigma^\tau}}\overline{\gf}^\tau=0$ with initial condition such that $\tilde{\gf}_\tau\vert_\Sigma=0$.

\begin{lma}\label{lemma:Dirichlet-YM}Let $\phi^D,\phi^N\in \ker \xd^{\star_\Sigma}$ be any pair of $1$-forms in $\Sigma$. Then there exists a solution $\gf\in \Omega^1 (\Sigma_\geps )$ of the YM condition belonging to the axial gauge fixing space in the bulk and whose boundary conditions belong to the divergence free gauge fixing space. In other words, there exists a solution of the following boundary value problem
\begin{gather*}
 \xd\star \xd\gf=0,\qquad \iota_{\partial \tau}\gf=0,\qquad \gf^D=\phi^D,\qquad \gf^N =\phi^N.
\end{gather*}
\end{lma}
\begin{proof}

Our aim is finding a one-parameter family $\overline{\gf}^\tau\in\Omega^1(\Sigma)$ such that
\begin{gather}\label{eqn:div-free}
 \xd^{\star_{\Sigma^\tau}}\frac{\xd}{\xd\tau}\overline{\gf}^\tau=0, \qquad \overline{\gf}^0=\phi^D,\qquad
 \left.\frac{\xd}{\xd\tau}\right\vert_{\tau=0}\overline{\gf}^\tau=\phi^N.
\end{gather}

We recall Moser's argument. Take the induced metric on $\Sigma^\tau$ and $\Sigma$, $\overline{h_{ij}}^\tau$ and $\overline{h_{ij}}$ respectively with corresponding volume forms, $\mu_{\Sigma^\tau}=\mu_\Sigma$, and $c_\tau$ such that $c_\tau \int_\Sigma\mu_{\Sigma^\tau} =\int_{\Sigma}\mu_\Sigma$. Recall that both volume forms $\mu_\Sigma$ and $(c_\tau\mu_{\Sigma^\tau})$ have the same cohomology class. Define the volume forms
\begin{gather*}
 \mu_{\tau}(t)=(1-t)\mu_\Sigma+t(c_\tau\mu_{\Sigma^\tau}),\qquad 0\leq t\leq 1 .
\end{gather*}
Then $\frac{\xd\mu_\tau(t)}{\xd t}=c_\tau\mu_{\Sigma^\tau}-\mu_\Sigma$. Define the $t$-dependent family of vector fields $Z_\tau(t)$ such that
\begin{gather*}
 \iota_{Z_\tau(t)}\mu_t = \mu_\Sigma - c_\tau\mu_{\Sigma^\tau} .
\end{gather*}
Let $\psi_\tau(t)$, $t\in[0,1]$ be the non-autonomous solution for $Z_\tau(t)$. Then
\begin{gather*}
(\psi_\tau(t))^*(c_\tau \mu_\tau(t))=\mu_\Sigma.
\end{gather*}
If we define a family of diffeomorphisms as $\psi_\tau\colon \Sigma\rightarrow\Sigma$, $ \tau\in [0,\geps]$, $\psi_\tau:=\psi_\tau(1)$, then
\begin{gather*}
\psi_\tau^*(c_\tau\mu_{\Sigma^\tau})=\mu_\Sigma, \qquad \psi_0={\rm id}_\Sigma.
\end{gather*}
Take
\begin{gather*}
 \frac{\xd}{\xd\tau} \overline{\gf}^\tau := \big(\psi_\tau^{-1}\big)^*\phi^N,\qquad \overline{\gf}^\tau := \phi^D+ \int_0^\tau \big(\psi_\tau^{-1}\big)^*\phi^N, \qquad \overline{\gf}_\tau\equiv 0.
\end{gather*}
Then, by definition, $ \big(X^0_\Sigma\big)^*\big( \frac{\xd}{\xd\tau}\big\vert_{\tau=0}\overline{\gf}\big) =\phi^N$, $ \big(X^0_\Sigma\big)^*\overline{\gf}^\tau=\phi^D$. Furthermore
$\xd\star_{\Sigma^\tau} \frac{\xd}{\xd\tau} \overline{\gf}^\tau = \xd\star_{\Sigma^\tau} \big(\psi_\tau^{-1}\big)^*\phi^N$, hence
\begin{gather*}
 c_\tau\psi_\tau^*\xd\star_{\Sigma^\tau} \frac{\xd}{\xd\tau} \overline{\gf}^\tau = c_\tau\psi_\tau^*\xd\star_{\Sigma^\tau}
 \big(\psi_\tau^{-1}\big)^*\phi^N = \xd c_\tau \psi_\tau^*\star_{\Sigma^\tau}\big(\psi_\tau^{-1}\big)^*\phi^N = \xd\star\phi^N=0 .
\end{gather*}
That is, $\xd\star_{\Sigma^\tau} \frac{\xd}{\xd\tau} \overline{\gf}^\tau=0 $.
\end{proof}

\begin{lma}\label{lemma:Dirichlet-YM2} Let $\phi^D\in \ker \xd^{\star_\Sigma}$ be any $1$-form in $\Sigma$. Then there exists a solution $\gf\in \Omega^1(\Sigma_\geps)$ of the YM condition belonging to the Lorentz gauge fixing space in the bulk and whose boundary conditions belong to the divergence free gauge fixing space with $\gf$ vanishing in the top boundary component of the cylinder $\Sigma_\geps$. In other words, there exists a solution of the following boundary value problem
\begin{gather*}
 \xd\star \xd\gf=0, \qquad \xd\star \gf=0, \qquad \gf^D=\phi^D,\qquad i_{\Sigma'}^*\gf = 0,
\end{gather*}
with $\Sigma'= X_\Sigma(\Sigma,\geps)$.
\end{lma}

\begin{proof}First we obtain a solution for axial gauge fixing, proceeding as in Lemma~\ref{lemma:Dirichlet-YM2}, but solving
\begin{gather*}%\label{eqn:div-free2}
 \xd^{\star_{\Sigma^\tau}}\frac{\xd}{\xd\tau}\overline{\gf}^\tau=0,\qquad \overline{\gf}^0=\phi^D,
\qquad \overline{\gf}^\geps=0,
\end{gather*}
instead of (\ref{eqn:div-free}). This $\gf$ solves YM condition~(\ref{eqn:YM-decomposition}). To obtain a Lorentz condition take a gauge translation
\begin{gather*}
 \tilde{\gf}=\tilde{\gf}^\tau+\tilde{\gf}_\tau \xd\tau=\gf + \xd f ,\qquad
 \tilde{\gf}_\tau=\frac{\partial f}{\partial \tau} ,\qquad \overline{\tilde{\gf}}^\tau=\overline{\gf}^\tau +\xd \overline{f}.
\end{gather*}
Lorentz gauge (\ref{eqn:div-decomposition}) can be written as
\begin{gather*}
 \frac{\partial}{\partial \tau}\overline{\tilde{\gf}}_\tau = -\xd^{\star_{\Sigma^\tau}}\overline{\tilde{\gf}}^\tau .
\end{gather*}
This implies
\begin{gather*}
 \frac{\partial^2 f}{\partial \tau^2} = -\xd^{\star_{\Sigma^\tau}}\overline{{\gf}}^\tau -\xd^{\star_{\Sigma^\tau}}\xd\overline{f},
\end{gather*}
which yield a Poisson equation for $f$ in $\Sigma_\geps$. This can be solved if we consider the suitable boundary conditions $\overline{f}\vert_\Sigma=0=\overline{f}\vert_{\Sigma'}$.
\end{proof}

\subsection*{Acknowledgments}
We thank the anonymous referees for contributions to improving a draft version of this paper. This work was partially supported by CONACYT project grant 259258, UNAM-DGAPA-PAPIIT project grant IN109415 and PRODEP project grant UMSNH-386.

\pdfbookmark[1]{References}{ref}
\LastPageEnding

\end{document}